\documentclass{LMCS}

\makeatletter
\@ifundefined{lhs2tex.lhs2tex.sty.read}%
  {\@namedef{lhs2tex.lhs2tex.sty.read}{}%
   \newcommand\SkipToFmtEnd{}%
   \newcommand\EndFmtInput{}%
   \long\def\SkipToFmtEnd#1\EndFmtInput{}%
  }\SkipToFmtEnd

\newcommand\ReadOnlyOnce[1]{\@ifundefined{#1}{\@namedef{#1}{}}\SkipToFmtEnd}
\usepackage{amstext}
\usepackage{amssymb}
\usepackage{stmaryrd}
\DeclareFontFamily{OT1}{cmtex}{}
\DeclareFontShape{OT1}{cmtex}{m}{n}
  {<5><6><7><8>cmtex8
   <9>cmtex9
   <10><10.95><12><14.4><17.28><20.74><24.88>cmtex10}{}
\DeclareFontShape{OT1}{cmtex}{m}{it}
  {<-> ssub * cmtt/m/it}{}

\DeclareFontShape{OT1}{cmtt}{bx}{n}
  {<5><6><7><8>cmtt8
   <9>cmbtt9
   <10><10.95><12><14.4><17.28><20.74><24.88>cmbtt10}{}
\DeclareFontShape{OT1}{cmtex}{bx}{n}
  {<-> ssub * cmtt/bx/n}{}

\newcommand{\Conid}[1]{\mathit{#1}}
\newcommand{\Varid}[1]{\mathit{#1}}
\newcommand{\anonymous}{\kern0.06em \vbox{\hrule\@width.5em}}

\usepackage{polytable}

\@ifundefined{mathindent}%
  {\newdimen\mathindent\mathindent\leftmargini}%
  {}%

\def\resethooks{%
  \global\let\SaveRestoreHook\empty
  \global\let\ColumnHook\empty}
\newcommand*{\savecolumns}[1][default]%
  {\g@addto@macro\SaveRestoreHook{\savecolumns[#1]}}
\newcommand*{\restorecolumns}[1][default]%
  {\g@addto@macro\SaveRestoreHook{\restorecolumns[#1]}}
\newcommand*{\aligncolumn}[2]%
  {\g@addto@macro\ColumnHook{\column{#1}{#2}}}

\resethooks

\newcommand{\onelinecommentchars}{\quad-{}- }
\newcommand{\commentbeginchars}{\enskip\{-}
\newcommand{\commentendchars}{-\}\enskip}

\newcommand{\visiblecomments}{%
  \let\onelinecomment=\onelinecommentchars
  \let\commentbegin=\commentbeginchars
  \let\commentend=\commentendchars}

\newcommand{\invisiblecomments}{%
  \let\onelinecomment=\empty
  \let\commentbegin=\empty
  \let\commentend=\empty}

\visiblecomments

\newlength{\blanklineskip}
\newcommand{\hsindent}[1]{\quad}
\let\hspre\empty
\let\hspost\empty

\EndFmtInput
\makeatother
\ReadOnlyOnce{polycode.fmt}%
\makeatletter

\newcommand{\hsnewpar}[1]%
  {{\parskip=0pt\parindent=0pt\par\vskip #1\noindent}}

\newcommand{\hscodestyle}{}

\newcommand{\sethscode}[1]%
  {\expandafter\let\expandafter\hscode\csname #1\endcsname
   \expandafter\let\expandafter\endhscode\csname end#1\endcsname}

  {\par\noindent
   \advance\leftskip\mathindent
   \hscodestyle
   \let\\=\@normalcr
   \let\hspre\(\let\hspost\)%
   \pboxed}%
  {\endpboxed\)%
   \par\noindent
   \ignorespacesafterend}

  {\hsnewpar\abovedisplayskip
   \advance\leftskip\mathindent
   \hscodestyle
   \let\hspre\(\let\hspost\)%
   \pboxed}%
  {\endpboxed%
   \hsnewpar\belowdisplayskip
   \ignorespacesafterend}

  {\hsnewpar\abovedisplayskip
   \advance\leftskip\mathindent
   \hscodestyle
   \let\\=\@normalcr
   \(\pboxed}%
  {\endpboxed\)%
   \hsnewpar\belowdisplayskip
   \ignorespacesafterend}

\newcommand{\plainhs}{\sethscode{plainhscode}}

\plainhs

  {\hsnewpar\abovedisplayskip
   \advance\leftskip\mathindent
   \hscodestyle
   \let\\=\@normalcr
   \(\parray}%
  {\endparray\)%
   \hsnewpar\belowdisplayskip
   \ignorespacesafterend}

  {\parray}{\endparray}

  {\(\parray}{\endparray\)}

\def\codeframewidth{\arrayrulewidth}
\RequirePackage{calc}

  {\parskip=\abovedisplayskip\par\noindent
   \hscodestyle
   \arrayrulewidth=\codeframewidth
   \tabular{@{}|p{\linewidth-2\arraycolsep-2\arrayrulewidth-2pt}|@{}}%
   \hline\framedhslinecorrect\\{-1.5ex}%
   \let\endoflinesave=\\
   \let\\=\@normalcr
   \(\pboxed}%
  {\endpboxed\)%
   \framedhslinecorrect\endoflinesave{.5ex}\hline
   \endtabular
   \parskip=\belowdisplayskip\par\noindent
   \ignorespacesafterend}

\newcommand{\framedhslinecorrect}[2]%
  {#1[#2]}

  {\(\def\column##1##2{}%
   \let\>\undefined\let\<\undefined\let\\\undefined
   \newcommand\>[1][]{}\newcommand\<[1][]{}\newcommand\\[1][]{}%
   \def\fromto##1##2##3{##3}%
   }{\) }%

  {\let\orighscode=\hscode
   \let\origendhscode=\endhscode
   \def\endhscode{\def\hscode{\endgroup\def\@currenvir{hscode}\\}\begingroup}
   \orighscode\def\hscode{\endgroup\def\@currenvir{hscode}}}%
  {\origendhscode
   \global\let\hscode=\orighscode
   \global\let\endhscode=\origendhscode}%

\makeatother
\EndFmtInput
\usepackage{amsmath, amssymb, booktabs, enumerate, graphicx, hyperref, mathpartir, tikz, url}

\newcommand{\skippres}[1]{}
\skippres{

\newtheorem{thm}{Theorem}[section]
\newtheorem{lem}[thmx]{Lemma}

}

\newcommand{\defqed}{\hfill$\diamond$\vspace{1ex}}
\newcommand{\exqed}{\hfill$\lhd$\vspace{1ex}}

\newcommand{\coq}{Coq}
\newcommand{\xml}{XML}
\newcommand{\json}{JSON}

\newcommand{\TRX}{TRX}

\newcommand{\LLk}{\textit{LL(k)}}
\newcommand{\LRk}{\textit{LR(k)}}
\newcommand{\LALRk}{\textit{LALR(k)}}

\newcommand{\comment}[1]{}

\newcommand{\GG}{\mathcal{G}}

\newcommand{\String}{\mathcal{S}}
\newcommand{\VV}{\mathcal{V}}
\newcommand{\nat}{\mathbb{N}}

\newcommand{\NT}{\VV_N}
\newcommand{\T}{\VV_T}

\newcommand{\ES}{\mathid{E}}
\newcommand{\PD}{\mathid{P_{exp}}}
\newcommand{\PS}{v_{\mathid{start}}}
\newcommand{\PT}{\mathid{P_{type}}}

\newcommand{\propSet}{\mathbb{P}}
\newcommand{\propEsym}{0}
\newcommand{\propNEsym}{{>0}}
\newcommand{\propNFsym}{\ge0}
\newcommand{\propFsym}{\bot}

\newcommand{\propSNE}{\propSet_{\propNEsym}}
\newcommand{\propSF}{\propSet_{\propFsym}}
\newcommand{\pegProp}[2]{#1 \in \propSet_{#2}}
\newcommand{\pegNProp}[2]{#1 \notin \propSet_{#2}}

\newcommand{\notPropE}[1]{\pegNProp{#1}{\propEsym}}
\newcommand{\propE}[1]{\pegProp{#1}{\propEsym}}
\newcommand{\propNE}[1]{\pegProp{#1}{\propNEsym}}
\newcommand{\propNF}[1]{\pegProp{#1}{\propNFsym}}
\newcommand{\propF}[1]{\pegProp{#1}{\propFsym}}

\newcommand{\WFset}{\mathid{WF}}
\newcommand{\WF}[1]{#1 \in \WFset}

\newcommand{\choiceOP}{/}
\newcommand{\choiceOPs}{\xx\choiceOP\xx}
\newcommand{\seqOP}{;}
\newcommand{\seqOPs}{\, \seqOP \,}
\newcommand{\plusOP}{+}

\newcommand{\emptye}{\epsilon}
\newcommand{\anychar}{[\cdot]}
\newcommand{\term}[1]{[#1]}
\newcommand{\lit}[1]{[\text{``$#1$''}]}
\newcommand{\xx}{\ \ }

\newcommand{\pseq}[2]{#1 \seqOP #2}
\newcommand{\pchoice}[2]{#1 \choiceOP #2}
\newcommand{\prep}[1]{#1*}
\newcommand{\pnot}[1]{!#1}
\newcommand{\pcoerce}[2]{#1 [\mapsto] #2}
\newcommand{\pcoerces}[2]{\pcoerce{#1 \xx}{\xx #2}}
\newcommand{\pcoercef}[2]{\pcoerces{#1}{{\, \lambda x\,.\,#2}}}
\newcommand{\dropResult}{[\sharp]}
\newcommand{\prange}[2]{[#1\!\!-\!\!#2]}
\newcommand{\pplus}[1]{{#1 \plusOP}}
\newcommand{\pand}[1]{\& #1}
\newcommand{\popt}[1]{{#1 ?}}

\newcommand{\mathid}[1]{\operatorname{#1}}
\newcommand{\imp}{\Rightarrow}
\newcommand{\set}[1]{\{#1\}}

\newcommand{\emptyList}{[]}
\newcommand{\cons}[2]{#1::#2}
\newcommand{\xs}{xs}
\newcommand{\xxs}{\cons{x}{\xs}}
\newcommand{\vs}{vs}
\newcommand{\vvs}{\cons{v}{\vs}}

\newcommand{\pegsemX}[5]{(#1, #2) \stackrel{\scriptscriptstyle #3}{#5} #4}
\newcommand{\pegsemx}[4]{\pegsemX{#1}{#2}{#3}{#4}{\leadsto}}
\newcommand{\notpegsemx}[3]{\pegsemX{#1}{#2}{}{#3}{\not\leadsto}}

\newcommand{\typei}[2]{#1\, :\, #2}

\newcommand{\pex}{\Delta}
\newcommand{\Pex}[1]{\Delta_{#1}}

\newcommand{\fail}{\bot}
\newcommand{\ok}[1]{\surd_{\!#1}}
\newcommand{\Ok}[2]{\surd^{\,#1}_{\,#2}}

\newcommand{\type}[1]{\mathrm{#1}}
\newcommand{\Type}{\type{Type}}
\newcommand{\TrueT}{\type{True}}
\newcommand{\Char}{\type{char}}
\newcommand{\StringT}{\type{string}}
\newcommand{\List}[1]{\type{list}\,#1}
\newcommand{\Option}[1]{\type{option}\,#1}
\newcommand{\NoneT}{\type{None}}
\newcommand{\SomeT}[1]{\type{Some}\;#1}

\newcommand{\ie}{\textit{i.e.}}
\newcommand{\resp}{\textit{resp.}}

\newcommand{\strlen}[1]{|#1|}
\newcommand{\subexp}{\sqsubseteq}

\newcommand{\NTid}[1]{\mathtt{#1}}
\newcommand{\NTnumber}{\NTid{number}}

\newcommand{\NTws}{\NTid{ws}}
\newcommand{\NTterm}{\NTid{term}}
\newcommand{\NTfactor}{\NTid{factor}}

\newcommand{\NTexpr}{\NTid{expr}}

\newcommand{\NTif}{\NTid{IF}}
\newcommand{\NTreserved}{\NTid{reserved}}
\newcommand{\NTident}{\NTid{identifier}}
\newcommand{\NTletter}{\NTid{letter}}
\newcommand{\NTstmt}{\NTid{stmt}}

\newcommand{\NCif}{\NTid{IF}}
\newcommand{\NCelse}{\NTid{ELSE}}
\newcommand{\NClp}{\NTid{(}}
\newcommand{\NCrp}{\NTid{)}}

\newcommand{\digListToNat}{\mathid{digListToNat}}
\newcommand{\coqid}[1]{\textit{#1}}

\theoremstyle{definition}

\def\doi{7 (2:18) 2011}
\lmcsheading%
{\doi}
{1--26}
{}
{}
{Jun.~14, 2010}
{Jun.~23, 2011}
{}   

\begin{document}

\title{TRX: A Formally Verified Parser Interpreter\rsuper*}

\author[A.~Koprowski]{Adam Koprowski}
\address{MLstate, Paris, France}
\email{Adam.Koprowski@mlstate.com, Henri.Binsztok@mlstate.com}

\author[H.~Binsztok]{Henri Binsztok}
\address{\vskip-6 pt}


\titlecomment{{\lsuper*}An extended abstract of this paper appeared in the Proceedings of the 19th
European Symposium on Programming \cite{KopBin10esop}.}

\begin{abstract}
Parsing is an important problem in computer science and yet
surprisingly little attention has been devoted to its formal verification.
In this paper, we present TRX: a parser interpreter formally developed
in the proof assistant \coq, capable of producing formally correct parsers. 
We are using parsing expression grammars (PEGs), a formalism essentially 
representing recursive descent parsing, which we consider an attractive 
alternative to context-free grammars (CFGs). From this formalization we 
can extract a parser for an arbitrary PEG grammar with the warranty of 
total correctness, i.e., the resulting parser is terminating and correct 
with respect to its grammar and the semantics of PEGs; both properties 
formally proven in \coq.
\end{abstract}

\maketitle

\section{Introduction}\label{sec:intro}

Parsing is of major interest in computer science.
Classically discovered by students as the first step in compilation, parsing is present in almost every program which performs data-manipulation. 

For instance, the Web is built on parsers. The HyperText Transfer Protocol (HTTP) is a parsed dialog between the client, or browser, and the server. This protocol transfers pages in HyperText Markup Language (HTML), which is also parsed by the browser.  When running web-applications, browsers interpret JavaScript programs which, again, begins with parsing.  Data exchange between browser(s) and server(s) uses languages or formats like \xml\ and \json. Even inside the server, several components (for instance the trio made of the HTTP server Apache, the PHP interpreter and the MySQL database) often manipulate programs and data dynamically; all require parsers.

Parsing is not limited to compilation or the Web: securing data flow entering a
network, signaling mobile communications, and manipulating domain specific languages
(DSL) all require a variety of parsers.

The most common approach to parsing is by means of \textit{parser generators}, which
take as input a grammar of some language and generate the source code of a parser for 
that language. They are usually based on regular expressions (REs) 
and context-free grammars (CFGs), the latter expressed in Backus-Naur Form (BNF) syntax.
They typically are able to deal with some subclass of context-free languages, the
popular subclasses including \LLk, \LRk\ and \LALRk\ grammars. Such grammars are usually
augmented with semantic actions that are used to produce a parse tree or an abstract
syntax tree (AST) of the input.

What about \emph{correctness} of such parsers? 
Yacc is the most widely used parser generator and a mature program and yet the reference book about this tool~\cite{lexyacc} devotes a whole section (``Bugs in Yacc'') to discuss common bugs in its distributions.
Furthermore, the code generated by such tools often contains huge parsing tables making it near impossible for manual inspection and/or verification. In the recent article about CompCert~\cite{Ler09}, an impressive project formally verifying a compiler for a large subset of C, the introduction starts with a question ``Can you trust your compiler?''.  Nevertheless, the formal verification starts on the level of the AST and does not concern the parser~\cite[Figure 1]{Ler09}.
Can you trust your parser?

\emph{Parsing expression grammars} (PEGs) \cite{For04} are an alternative to CFGs,
that have recently been gaining popularity. In contrast to CFGs they are unambiguous and allow 
easy integration of lexical analysis into the parsing phase. Their implementation is easy, 
as PEGs are essentially a declarative way of specifying recursive descent parsers \cite{Bur75}. 
With their backtracking and unlimited look-ahead capabilities they are expressive enough to 
cover all \LLk\ and \LRk\ languages as well as some non-context-free ones. However, recursive 
descent parsing of grammars that are not \LLk\ may require exponential time. A solution to 
that problem is to use memoization giving rise to \emph{packrat parsing} and ensuring linear 
time complexity at the price of higher memory consumption~\cite{AhoUll72,For02,For02mth}.
It is not easy to support (indirect) left-recursive rules in PEGs, as they lead to 
non-terminating parsers \cite{WarEA08}.

In this paper we present \TRX: a PEG-based parser interpreter \emph{formally developed} in the
proof assistant \coq\ \cite{Coq,BerCas04}. As a result, expressing a grammar in \coq\
allows one, via its extraction capabilities \cite{Let08}, to obtain a parser for this grammar
with \emph{total correctness guarantees}. That means that the resulting parser is terminating 
and correct with respect to its grammar and the semantics of PEGs; both of those 
properties formally proved in \coq. Moreover every definition and theorem presented in 
this paper has been expressed and verified in \coq.

Our emphasis is on the \emph{practicality} of such a tool. We perform two case studies: on a simple
XML format but also on the full grammar of the Java language. We present benchmarks indicating
that the performance of obtained parsers is reasonable. We also sketch ideas on how it can be
improved further, as well as how TRX could be extended into a tool of its own, freeing its 
users from any kind of interaction with Coq and broadening its applicability.

This work was carried out in the context of improving safety and security of OPA (One Pot Application):
an integrated platform for web development \cite{OPA}. As mentioned above parsing is of
uttermost importance for web-applications and \TRX\ is one of the components in the 
OPA platform.

The remainder of this paper is organized as follows. We introduce PEGs in Section~\ref{sec:pegs} 
and in Section~\ref{sec:pegs-actions} we extend them with semantic actions.
Section~\ref{sec:pegs-wf} describes a method for checking that there is no (indirect)
left recursion in a grammar, a result ensuring that parsing will terminate.
Section~\ref{sec:pegs-int} reports on our experience with putting the ideas of the preceding
sections into practice and implementing a formally correct parser interpreter in \coq.
Section~\ref{sec:pegs-ex} is devoted to a practical evaluation of this interpreter and contains 
case studies of extracting \xml\ and Java parsers from it, presenting a benchmark of
\TRX\ against other parser generators and giving an account of our experience with 
extraction. We discuss related work in Section~\ref{sec:relwork}, present ideas for
extensions and future work in Section~\ref{sec:discussion} and conclude in Section~\ref{sec:concl}.

\section{Parsing Expression Grammars (PEGs)}\label{sec:pegs}

The content of this section is a different presentation of the 
results by Ford~\cite{For04}. For more details we refer to the original article.
For a general overview of parsing we refer to, for instance, Aho, Seti \& Ullman~\cite{AhoEA86}.

PEGs are a formalism for parsing that is an interesting alternative to 
CFGs. We will formally introduce them along with their semantics in
Section~\ref{sec:peg_def}. PEGs are gaining popularity recently due to 
their ease of implementation and some general desirable properties that 
we will sketch in Section~\ref{sec:peg_vs_cfg}, while comparing them to CFGs.

\subsection{Definition of PEGs}\label{sec:peg_def}

\begin{defi}[Parsing expressions]\label{pexp}
We introduce a set of \emph{parsing expressions}, $\pex$, over a finite
set of terminals $\T$ and a finite set of non-terminals $\NT$. 
We denote the set of strings as $\String$ and a string $s \in \String$ 
is a list of terminals $\T$.
The inductive definition of $\pex$ is given in Figure~\ref{peg-exp-fig}.
\defqed
\end{defi}

\begin{figure}[t!]
\begin{center}
\[
\begin{array}{r@{\;\;}l@{\;\;}rl@{\;\;}ll@{\;\;}rl}
\pex ::=
     & \emptye            & \text{empty expr.}                               & &
\mid & \pchoice{e_1}{e_2} & \text{a \emph{prioritized} choice}               & (e_1, e_2 \in \pex) \\
\mid & \anychar           & \text{any character}                             & &
\mid & \prep{e}           & \text{a $\ge 0$ \textit{greedy} repetition} & (e \in \pex) \\
\mid & \term{a}           & \text{a terminal}                                & (a \in \T) &
\mid & \pplus{e}          & \text{a $\ge 1$ \textit{greedy} repetition}  & (e \in \pex) \\
\mid & \lit{s}            & \text{a literal}                                 & (s \in \String) &
\mid & \popt{e}           & \text{an optional expression}                    & (e \in \pex) \\
\mid & \prange{a}{z}      & \text{a range}                                   & (a, z \in \T) &
\mid & \pnot{e}           & \text{a not-predicate}                           & (e \in \pex) \\
\mid & A                  & \text{a non-terminal}                            & (A \in \NT) &
\mid & \pand{e}           & \text{an and-predicate}                          & (e \in \pex) \\
\mid & \pseq{e_1}{e_2}    & \text{a sequence}                                & (e_1, e_2 \in \pex)
\end{array}
\]
\end{center}
\caption{Parsing expressions}
\label{peg-exp-fig}
\end{figure}

Later on we will present the formal semantics but 
for now we informally describe the language expressed by such parsing expressions.

\begin{iteMize}{$\bullet$}
  \item \emph{Empty expression} $\emptye$ always succeeds without consuming 
    any input.
  \item \emph{Any-character} $\anychar$, a \emph{terminal} $\term{a}$ and
    a \emph{range} $[a-z]$ all consume a single terminal from the input but they expect
    it to be, respectively: an arbitrary terminal, precisely $a$ and in the range
    between $a$ and $z$.
  \item \emph{Literal} $\lit{s}$ reads a string (\ie, a sequence of terminals) $s$ from
    the input.
  \item Parsing a \emph{non-terminal} $A$ amounts to parsing the expression 
    defining $A$.
  \item A \emph{sequence} $\pseq{e_1}{e_2}$ expects an input conforming to 
    $e_1$ followed by an input conforming to $e_2$.
  \item A \emph{choice} $\pchoice{e_1}{e_2}$ expresses a \emph{prioritized} 
    choice between $e_1$ and $e_2$. This means that $e_2$ will be tried only if
    $e_1$ fails.
  \item A \emph{zero-or-more (\resp\ one-or-more) repetition} 
    $\prep{e}$ (\resp\ $\pplus{e}$) consumes zero-or-more (\resp\ one-or-more) 
    repetitions of $e$ from the input. Those operators are \emph{greedy}, \ie,
    the longest match in the input, conforming to $e$, will be consumed.
  \item An \emph{and-predicate (\resp\ not-predicate)} $\pand{e}$ (\resp\ $\pnot{e}$) 
    succeeds only if the input conforms to $e$ (\resp\ does not conform to $e$) but does 
    not consume any input.
\end{iteMize}\medskip

\noindent We now define PEGs, which are essentially a finite set of
non-terminals, also referred to as \emph{productions}, with their
corresponding parsing expressions.

\begin{defi}[Parsing Expressions Grammar (PEG)]\label{peg}
  A parsing expressions grammar (PEG), $\GG$, is a tuple
$(\T, \NT, \PD, \PS)$, where:
\begin{iteMize}{$\bullet$}
  \item $\T$ is a finite set of terminals,
  \item $\NT$ is a finite set of non-terminals,
  \item $\PD$ is the interpretation of the productions, \ie, $\PD : \NT \to \pex$ and
  \item $\PS$ is the start production, $\PS \in \NT$.  \defqed
\end{iteMize}
\end{defi}

We will now present the formal semantics of PEGs. 
The semantics is given by means of tuples $\pegsemx{e}{s}{m}{r}$, which indicate 
that parsing expression $e \in \pex$ applied on a string $s \in \String$ gives,
in $m$ steps, the result $r$, where $r$ is either $\fail$, denoting that parsing failed, 
or $\ok{s'}$, indicating that parsing succeeded and $s'$ is what remains to be 
parsed. We will drop the $m$ annotation whenever irrelevant.

\begin{figure}[t!]
\begin{center}
\[
\begin{array}{c@{\qquad}c@{\qquad}c}
\inferrule{ }{\pegsemx{\emptye}{s}{1}{\ok{s}}} &
\inferrule{\pegsemx{\PD(A)}{s}{n}{r}}{\pegsemx{A}{s}{n+1}{r}} &
\inferrule{ }{\pegsemx{\anychar}{\xxs}{1}{\ok{\xs}}} \\[15pt]
\inferrule{ }{\pegsemx{\anychar}{\emptyList}{1}{\fail}} &
\inferrule{ }{\pegsemx{\term{x}}{\xxs}{1}{\ok{\xs}}} &
\inferrule{ }{\pegsemx{\term{x}}{\emptyList}{1}{\fail}} \\[15pt]
\inferrule{x \neq y}{\pegsemx{\term{y}}{\xxs}{1}{\fail}} &
\inferrule{\pegsemx{e}{s}{m}{\fail}}{\pegsemx{\pnot{e}}{s}{m+1}{\ok{s}}} &
\inferrule{\pegsemx{e}{s}{m}{\ok{s'}}}{\pegsemx{\pnot{e}}{s}{m+1}{\fail}} \\[15pt]
\inferrule{\pegsemx{e_1}{s}{m}{\fail}}{\pegsemx{\pseq{e_1}{e_2}}{s}{m+1}{\fail}} &
\inferrule{\pegsemx{e_1}{s}{m}{\ok{s'}} \\ \pegsemx{e_2}{s'}{n}{r}}{\pegsemx{\pseq{e_1}{e_2}}{s}{m+n+1}{r}} &
\inferrule{\pegsemx{e_1}{s}{m}{\fail} \\ \pegsemx{e_2}{s}{n}{r}}{\pegsemx{\pchoice{e_1}{e_2}}{s}{m+n+1}{r}} \\[15pt]
\inferrule{\pegsemx{e_1}{s}{m}{\ok{s'}}}{\pegsemx{\pchoice{e_1}{e_2}}{s}{m+1}{\ok{s'}}} &
\inferrule{\pegsemx{e}{s}{m}{\ok{s'}} \\ \pegsemx{\prep{e}}{s'}{n}{\ok{s''}}}{\pegsemx{\prep{e}}{s}{m+n+1}{\ok{s''}}} &
\inferrule{\pegsemx{e}{s}{m}{\fail}}{\pegsemx{\prep{e}}{s}{m+1}{\ok{s}}}
\end{array}
\]
\end{center}
\caption{Formal semantics of PEGs}
\label{peg-sem-fig}
\end{figure}

The complete semantics is presented in Figure~\ref{peg-sem-fig}. Please note
that the following operators from Definition~\ref{pexp} can be derived and
therefore are not included in the semantics:
\[
\begin{array}{r @{\ \ ::=\ \ } l @{\qquad\qquad} r @{\ \ ::=\ \ } l @{\qquad\qquad} r @{\ \ ::=\ \ } l}
\prange{a}{z} & \term{a} \choiceOPs \ldots \choiceOPs \term{z} &
\pplus{e}     & \pseq{e}{\prep{e}} &
\pand{e}      & \pnot{\pnot{e}} \\
\lit{s}       & \term{s_0} \seqOPs \ldots \seqOPs \term{s_n} &
\popt{e}      & \pchoice{e}{\emptye}
\end{array}
\]


\subsection{CFGs vs PEGs}\label{sec:peg_vs_cfg}

The main differences between PEGs and CFGs are the following:
\begin{iteMize}{$\bullet$}
  \item the choice operator, $\pchoice{e_1}{e_2}$, is \emph{prioritized}, \ie, 
    $e_2$ is tried only if $e_1$ fails;
  \item the repetition operators, $\prep{e}$ and $\pplus{e}$, are \emph{greedy},
    which allows to easily express ``longest-match'' parsing, which is almost
    always desired;
  \item \emph{syntactic predicates} \cite{ParQuo94}, $\pand{e}$ and $\pnot{e}$, both of which consume
    no input and succeed if $e$, respectively, succeeds or fails. This effectively
    provides an \emph{unlimited look-ahead} and, in combination with choice, 
    limited \emph{backtracking} capabilities.
\end{iteMize}\medskip

\noindent An important consequence of the choice and repetition operators being deterministic 
(choice being prioritized and repetition greedy) is the fact that PEGs are 
\emph{unambiguous}. We will see a formal proof of that in Theorem~\ref{thm:peg_unambig}.
This makes them unfit for processing natural languages, 
but is a much desired property when it comes to grammars for programming 
languages. 

Another important consequence is ease of implementation. Efficient algorithms are
known only for certain subclasses of CFGs and they tend to be rather complicated. 
PEGs are essentially a declarative way of specifying \emph{recursive descent parsers}~\cite{Bur75}
and performing this type of parsing for PEGs is straightforward (more
on that in Section~\ref{sec:pegs-int}). By using the technique of
\emph{packrat parsing} \cite{AhoUll72,For02}, \ie, essentially adding memoization to the
recursive descent parser, one obtains parsers with linear time complexity guarantees. 
The downside of this approach is high memory requirements: the worst-time space complexity 
of PEG parsing is linear in the size of the input, but with packrat parsing the
constant of this correlation can be very high. For instance Ford reports on a factor 
of around 700 for a parser of Java \cite{For02}.

CFGs work hand-in-hand with REs. The \emph{lexical analysis}, breaking up the input 
into tokens, is performed with REs. Such tokens are subject to \emph{syntactical analysis},
which is executed with CFGs. This split into two phases is not necessary with PEGs, 
as they make it possible to easily express both lexical and syntactical rules with
a single formalism. We will see that in the following example.


\begin{exa}[PEG for simple mathematical expressions]\label{math_peg}
 Consider a PEG for simple mathematical expressions over 5 non-terminals:
  $\NT ::= \{\NTws, \NTnumber, \NTterm,$ $\NTfactor, \NTexpr\}$
with the following productions ($\PD$ function from Definition~\ref{peg}):
\[
\begin{array}{r @{\ \ ::=\ \ } l}
\NTws     & \prep{(
               \term{{\text{\textvisiblespace }}} 
            \choiceOPs
                \term{{\backslash t}}
            )} \\
\NTnumber & \pplus{\prange{0}{9}} \\
\NTterm   & \NTws \xx \NTnumber \xx \NTws \choiceOPs
            \NTws \xx \term{{(}} \xx \NTexpr \xx \term{{)}} \xx \NTws \\
\NTfactor & \NTterm \xx \term{{*}} \xx \NTfactor \choiceOPs \NTterm \\
\NTexpr   & \NTfactor \xx \term{{+}} \xx \NTexpr \choiceOPs \NTfactor \\
\end{array}
\]
Please note that in this and all the following examples we write the
sequence operator $\pseq{e_1}{e_2}$ implicitly as $e_1\ e_2$. The starting 
production is $\PS ::= \NTexpr$.

First, let us note that lexical analysis is incorporated into this grammar by means
of the $\NTws$ production which consumes all white-space from the beginning of the input.
Allowing white-space between ``tokens'' of the grammar comes down to placing 
the call to this production around the terminals of the grammar. If one does not 
like to clutter the grammar with those additional calls then a simple solution is
to re-factor all terminals into separate productions, which consume not only the terminal
itself but also all white-space around it.

Another important observation is that we made addition (and also multiplication) right-associative.
If we were to make it, as usual, left-associative, by replacing the rule 
for $\NTexpr$ with:

\[
 \NTexpr \ ::= \ \NTexpr \xx \term{{+}} \xx \NTfactor \choiceOPs \NTfactor
\]

\noindent then we get a grammar that is left-recursive. Left-recursion (also indirect or mutual)
is problematic as it leads to non-terminating parsers. We will come back to this issue
in Section~\ref{sec:pegs-wf}. 
\exqed
\end{exa}

PEGs can also easily deal with some common idioms often encountered in practical 
grammars of programming languages, which pose a lot of difficulty for CFGs, such
as modular way of handling reserved words of a language and a ``dangling'' else problem 
--- we present them on two examples and refer for more details to Ford~\cite[Chapter 2.4]{For02mth}.

\begin{exa}[Reserved words]
One of the difficulties in tokenization is that virtually every programming 
language has a list of \emph{reserved words}, which should not be accepted 
as identifiers. PEGs allow an elegant pattern to deal with this problem:
\[
\begin{array}{r @{\ \ ::=\ \ } l}
\NTident    & \pnot{\NTreserved} \xx \pplus{\NTletter} \xx \NTws \\
\NTreserved & \NTif \xx \choiceOPs \ldots \\
\NTif       & \lit{if} \xx \pnot{\NTletter} \xx \NTws
\end{array}
\]
The rule $\NTident$ for identifiers reads a non-empty list of letters 
but only after checking, with the not-predicate, that there is no 
reserved word at this position. The rules for the reserved words 
ensure that it is not followed by a letter (``$\operatorname{ifs}$''
is a valid identifier) and consume all the following white space.
In this example we only presented a single reserved word ``$\operatorname{if}$''
but adding a new word requires only adding a rule similar to $\NTif$ 
and extending the choice in $\NTreserved$.
\exqed
\end{exa}

\begin{exa}[``Dangling'' else]
Consider the following part of a CFG for the C language:
\[
\begin{array}{r @{\ \ ::=\ \ } l}
\NTstmt                  & \NCif \xx \NClp \xx \NTexpr \xx \NCrp \xx \NTstmt \\
\multicolumn{1}{r}{\mid} & \NCif \xx \NClp \xx \NTexpr \xx \NCrp \xx \NTstmt \xx \NCelse \xx \NTstmt \\
\multicolumn{1}{r}{\mid} & \ldots
\end{array}
\]
According to this grammar there are two possible readings of a statement
\[
  \mathrm{if} \xx (e_1) \xx \mathrm{if} \xx (e_2) \xx s_1 \xx \mathrm{else} \xx s_2
\]
as the ``$\mathrm{else} \xx s_2$'' branch can be associated either with the 
outer or the inner $\mathrm{if}$. The desired way to resolve this 
ambiguity is usually to bind this $\mathrm{else}$ to the innermost construct.
This is exactly the behavior that we get by converting this CFG to a PEG
by replacing the symmetrical choice operator ``$\mid$'' of CFGs with the
prioritized choice of PEGs ``$\choiceOP$''.
\exqed
\end{exa}

\section{Extending PEGs with Semantic Actions}\label{sec:pegs-actions}

\subsection{XPEGs: Extended PEGs}

  In the previous section we introduced parsing expressions, which can be
used to specify which strings belong to the grammar under consideration.
However the role of a parser is not merely to recognize whether an input is
correct or not but also, given a correct input, to compute its representation
in some structured form. This is typically done by extending grammar expressions
with \emph{semantic values}, which are a representation of the result of parsing 
this expression on (some) input and by extending a grammar with \emph{semantic
actions}, which are functions used to produce and manipulate the semantic 
values. Typically a semantic value associated with an expression will be
its parse tree so that parsing a correct input will give a \emph{parse tree} of this
input. For programming languages such parse tree would represent the AST of 
the language.

In order to deal with this extension we will replace the simple type of parsing 
expressions $\pex$ with a family of types $\Pex{\alpha}$, where the index $\alpha$
is a type of the semantic value associated with the expression.
We also compositionally define default semantic values for all types of 
expressions and introduce a new construct: coercion, $\pcoerce{e}{f}$, which 
converts a semantic value $v$ associated with $e$ to $f(v)$. 

Borrowing notations from \coq\ we will use the following types:
\begin{iteMize}{$\bullet$}
  \item $\Type$ is the universe of types.
  \item $\TrueT$ is the singleton type with a single value $I$.
  \item $\Char$ is the type of machine characters. It corresponds to the type of 
    terminals $\T$, which in concrete parsers will always be instantiated
    to $\Char$.
  \item $\List{\alpha}$ is the type of lists of elements of $\alpha$ for any 
    type $\alpha$. Also $\StringT ::= \List{\Char}$.
  \item $\alpha_1 * \ldots * \alpha_n$ is the type of $n$-tuples of elements $(a_1, \ldots, a_n)$ with 
    $a_1 \in \alpha_1, \ldots, a_n \in \alpha_n$ for any types $\alpha_1, \ldots, \alpha_n$.
    If $v$ is an $n$-tuple then $v_i$ is its $i$'th projection.
  \item $\Option{\alpha}$ is the type optionally holding a value of type $\alpha$,
    with two constructors $\NoneT$ and $\SomeT{v}$ with $\typei{v}{\alpha}$.
\end{iteMize}

\begin{defi}[Parsing expressions with semantic values]\label{PExp}
We introduce a set of \emph{parsing expressions with semantic values}, $\Pex{\alpha}$, 
as an inductive family indexed by the type $\alpha$ of semantic values of
an expression. The typing rules for $\Pex{\alpha}$ are given in 
Figure~\ref{peg-exp-prod-fig}. 
\defqed
\end{defi}

Note that for the choice operator $\pchoice{e_1}{e_2}$ the types of semantic values 
of $e_1$ and $e_2$ must match, which will sometimes require use of the coercion 
operator $\pcoerce{e}{f}$.

Let us again see the derived operators and their types, as we need to insert
a few coercions:
\[
\begin{array}{r @{\ :\ } l @{\ \ ::=\ \ } l}
\prange{a}{z} & \Pex{\Char}           & \term{a} \choiceOPs \ldots \choiceOPs \term{z} \\
\lit{s}       & \Pex{\StringT}        & \pcoerces{\term{s_0} \seqOPs \ldots \seqOPs \term{s_n}}{\operatorname{tuple2str}} \\
\pplus{e}     & \Pex{\List{\alpha}}   & \pcoercef{\pseq{e}{\prep{e}}}{\cons{x_1}{x_2}} \\
\popt{e}      & \Pex{\Option{\alpha}} & \pcoercef{\makebox[0.2cm][l]{$e$}}{\SomeT{x}} \\
\multicolumn{2}{r}{\choiceOP}         & \pcoercef{\makebox[0.2cm][l]{$\emptye$}}{\NoneT} \\
\pand{e}      & \Pex{\TrueT}          & \pnot{\pnot{e}}
\end{array}
\]
where $\operatorname{tuple2str}(c_1, \ldots, c_n) = [c_1; \ldots; c_n]$.

\begin{figure}[t!]
\begin{center}
\[
\begin{array}{ccc}
\inferrule{ }{\typei{\emptye}{\Pex{\TrueT}}} &
\inferrule{ }{\typei{\anychar}{\Pex{\Char}}} &
\inferrule{a \in \T}{\typei{\term{a}}{\Pex{\Char}}} \\[15pt]
\inferrule{A \in \NT}{\typei{A}{\Pex{\PT(A)}}} &
\qquad \inferrule{\typei{e_1}{\Pex{\alpha}} \\ \typei{e_2}{\Pex{\beta}}}{\typei{\pseq{e_1}{e_2}}{\Pex{\alpha * \beta}}} &
\qquad \inferrule{\typei{e_1}{\Pex{\alpha}} \\ \typei{e_2}{\Pex{\alpha}}}{\typei{\pchoice{e_1}{e_2}}{\Pex{\alpha}}} \\[15pt]
\inferrule{\typei{e}{\Pex{\alpha}}}{\typei{\prep{e}}{\Pex{\List{\alpha}}}} &
\inferrule{\typei{e}{\Pex{\alpha}}}{\typei{\pnot{e}}{\Pex{\TrueT}}} &
\inferrule{\typei{e}{\Pex{\alpha}} \\ \typei{f}{\alpha \to \beta}}{\typei{\pcoerce{e}{f}}{\Pex{\beta}}}
\end{array}
\]
\end{center}
\caption{Typing rules for parsing expressions with semantic actions}
\label{peg-exp-prod-fig}
\end{figure}

  The definition of an extended parsing expression grammar (XPEG) is 
as expected (compare with Definition~\ref{pexp}).

\begin{defi}[Extended Parsing Expressions Grammar (XPEG)]\label{pegs-prod-peg}
  An extended parsing expressions grammar (XPEG), $\GG$, is a tuple
$(\T, \NT, \PT,$ $ \PD, \PS)$, where:
\begin{iteMize}{$\bullet$}
  \item $\T$ is a finite set of terminals,
  \item $\NT$ is a finite set of non-terminals,
  \item $\PT : \NT \to \Type$ is a function that gives types of semantic 
    values of all productions.
  \item $\PD$ is the interpretation of the productions of the grammar, \ie,
    $\PD : \forall_{A : \NT} \Pex{\PT(A)}$ and
  \item $\PS$ is the start production, $\PS \in \NT$. \defqed
\end{iteMize}
\end{defi}

  We extended the semantics of PEGs from Figure~\ref{peg-sem-fig} to semantics
of XPEGs in Figure~\ref{peg-sem-prod-fig}.
 
\newcommand{\rsucc}[1]{#1_{\surd}}
\newcommand{\rfail}[1]{#1_{\bot}}
\newcommand{\repS}{\rsucc{rep}}
\newcommand{\repF}{\rfail{rep}}

\begin{figure}[t!]
\begin{center}
\[
\begin{array}{c@{\ }c@{\ }c}
\inferrule{ }{\pegsemx{\emptye}{s}{1}{\Ok{I}{s}}} &
\inferrule{\pegsemx{\PD(A)}{s}{m}{r}}{\pegsemx{A}{s}{m+1}{r}} &
\inferrule{ }{\pegsemx{\anychar}{\xxs}{1}{\Ok{x}{\xs}}} \\[15pt]
\inferrule{ }{\pegsemx{\anychar}{\emptyList}{1}{\fail}} &
\inferrule{\pegsemx{e_1}{s}{m}{\fail} \\ \pegsemx{e_2}{s}{n}{r}}{\pegsemx{\pchoice{e_1}{e_2}}{s}{m+n+1}{r}} &
\inferrule{\pegsemx{e_1}{s}{m}{\Ok{v}{s'}}}{\pegsemx{\pchoice{e_1}{e_2}}{s}{m+1}{\Ok{v}{s'}}} \\[15pt]
\inferrule{ }{\pegsemx{\term{x}}{\xxs}{1}{\Ok{x}{\xs}}} &
\inferrule{ }{\pegsemx{\term{x}}{\emptyList}{1}{\fail}} &
\inferrule{x \neq y}{\pegsemx{\term{y}}{\xxs}{1}{\fail}} \\[15pt]
\inferrule{\pegsemx{e_1}{s}{m}{\Ok{v_1}{s'}} \\ \pegsemx{e_2}{s'}{n}{\fail}}{\pegsemx{\pseq{e_1}{e_2}}{s}{m+n+1}{\fail}} &
\inferrule{\pegsemx{e_1}{s}{m}{\Ok{v_1}{s'}} \\ \pegsemx{e_2}{s'}{n}{\Ok{v_2}{s''}}}{\pegsemx{\pseq{e_1}{e_2}}{s}{m+n+1}{\Ok{(v_1,v_2)}{s''}}} &
\inferrule{\pegsemx{e_1}{s}{m}{\fail}}{\pegsemx{\pseq{e_1}{e_2}}{s}{m+1}{\fail}} \\[15pt]
\inferrule{\pegsemx{e}{s}{m}{\fail}}{\pegsemx{\prep{e}}{s}{m+1}{\Ok{\emptyList}{s}}} &
\inferrule{\pegsemx{e}{s}{m}{\Ok{v}{s'}} \\ \pegsemx{\prep{e}}{s'}{n}{\Ok{\vs}{s''}}}{\pegsemx{\prep{e}}{s}{m+n+1}{\Ok{\vvs}{s''}}} &
\inferrule{\pegsemx{e}{s}{m}{\fail}}{\pegsemx{\pnot{e}}{s}{m+1}{\Ok{I}{s}}} \\[15pt]
\inferrule{\pegsemx{e}{s}{m}{\Ok{v}{s'}}}{\pegsemx{\pnot{e}}{s}{m+1}{\fail}} &
\inferrule{\pegsemx{e}{s}{m}{\Ok{v}{s'}}}{\pegsemx{\pcoerce{e}{f}}{s}{m+1}{\Ok{f(v)}{s'}}} &
\inferrule{\pegsemx{e}{s}{m}{\fail}}{\pegsemx{\pcoerce{e}{f}}{s}{m+1}{\fail}}
\end{array}
\]
\end{center}
\caption{Formal semantics of XPEGs with semantic actions.}
\label{peg-sem-prod-fig}
\end{figure}

\begin{exa}[Simple mathematical expressions ctd.]\label{math_xpeg}
 Let us extend the grammar from Example~\ref{math_peg} with semantic
actions. The grammar expressed mathematical expressions and 
we attach semantic actions evaluating those expressions, hence obtaining
a very simple calculator.

  It often happens that we want to ignore the semantic value
attached to an expression. This can be accomplished by coercing this
value to $I$, which we will abbreviate by $e\dropResult ::= \pcoercef{e}{I}$.
\[
\begin{array}{r @{\ \ ::=\ \ } l l}
\NTws     & \prep{(
               \term{{\text{\textvisiblespace }}} 
            \choiceOPs
                \term{{\backslash t}}
            )} & \xx\dropResult \\
\NTnumber & \pplus{\prange{0}{9}}          & \pcoerces{}{\digListToNat} \\
\NTterm   & \NTws \xx \NTnumber \xx \NTws  & \pcoercef{}{x_2} \\
\multicolumn{1}{r}{\choiceOPs} & \NTws \xx \term{{(}} \xx \NTexpr \xx \term{{)}} \xx \NTws & \pcoercef{}{x_3} \\
\NTfactor & \NTterm \xx \term{{*}} \xx \NTfactor & \pcoercef{}{x_1 * x_3} \\
\multicolumn{1}{r}{\choiceOPs} & \NTterm \\
\NTexpr   & \NTfactor \xx \term{{+}} \xx \NTexpr & \pcoercef{}{x_1 + x_3} \\
\multicolumn{1}{r}{\choiceOPs} & \NTfactor \\
\end{array}
\]
where $\digListToNat$ converts a list of digits to their decimal representation
and $x_i$ in the productions is the $i$-th projection of the vector of values $x$,
resulting from parsing a sequence.

This grammar will associate, as expected, the semantical value $36$ with
the string ``\texttt{(1+2) * (3 * 4)}''. Of course in practice instead of
evaluating the expression we would usually write semantic actions to build 
a parse tree of the expression for later processing.
\exqed
\end{exa}

\subsection{Meta-properties of (X)PEGs}

Now we will present some results concerning semantics of (X)PEGs. They are all
variants of results obtained by Ford \cite{For04}, only now we extend them to XPEGs.
First we prove that, as expected, the parsing only 
consumes a prefix of a string.

\begin{thm}\label{thm:peg_prefix}
If $\pegsemx{e}{s}{m}{\Ok{v}{s'}}$ then $s'$ is a suffix of $s$.
\end{thm}
\begin{proof}
Induction on the derivation of $\pegsemx{e}{s}{m}{\Ok{v}{s'}}$
using transitivity of the prefix property for 
sequence and repetition cases.
\end{proof}

As mentioned earlier, (X)PEGs are unambiguous:

\begin{thm}\label{thm:peg_unambig}
If $\pegsemx{e}{s}{m_1}{r_1}$ and $\pegsemx{e}{s}{m_2}{r_2}$ then $m_1 = m_2$ and $r_1 = r_2$.
\end{thm}
\begin{proof}
Induction on the derivation $\pegsemx{e}{s}{m_1}{r_1}$ followed by inversion
of $\pegsemx{e}{s}{m_2}{r_2}$. All cases immediate from the semantics of XPEGs.
\end{proof}

We wrap up this section with a simple property about the repetition operator, 
that we will need later on. It states that the semantics of a repetition 
expression $\prep{e}$ is not defined if $e$ succeeds without consuming any input.

\begin{lem}\label{thm:peg_loop_cond}
If $\pegsemx{e}{s}{m}{\Ok{v}{s}}$ then $\notpegsemx{\prep{e}}{s}{r}$ for all $r$.
\end{lem}
\begin{proof}
Assume $\pegsemx{e}{s}{m}{\Ok{v}{s}}$ and $\pegsemx{\prep{e}}{s}{n}{\Ok{vs}{s'}}$ 
for some $n$, $vs$ and $s'$ (we cannot have $\pegsemx{\prep{e}}{s}{n}{\fail}$ as $\prep{e}$ 
never fails). By the first rule for repetition $\pegsemx{\prep{e}}{s}{m+n+1}{\Ok{v::vs}{s'}}$,
which contradicts the second assumption by Theorem~\ref{thm:peg_unambig}.
\end{proof}

\section{Well-formedness of PEGs}\label{sec:pegs-wf}

We want to guarantee \emph{total correctness} for generated parsers,
meaning they must be \emph{correct} (with respect to PEGs semantics)
and \emph{terminating}. In this section we focus on the latter problem.
Throughout this section we assume a fixed PEG $\GG$.

\subsection{Termination problem for XPEGs}

Ensuring termination of a PEG parser essentially comes down to
two problems:
\begin{iteMize}{$\bullet$}
  \item termination of all semantic actions in $\GG$ and
  \item completeness of $\GG$ with respect to PEGs semantics.
\end{iteMize}

As for the first problem it means that all $f$ functions used
in coercion operators $\pcoerce{e}{f}$ in $\GG$, must be terminating.
We are going to express PEGs completely in \coq\ (more on that in 
Section~\ref{sec:pegs-int}) so for our application we get this property 
for free, as all \coq\ functions are total (hence terminating).

Concerning the latter problem, we must ensure that the grammar $\GG$
under consideration is \emph{complete}, \ie, it either succeeds 
or fails on all input strings. The only potential source of incompleteness 
of $\GG$ is (mutual) \emph{left-recursion} in the grammar. 

We already hinted at this problem in Example~\ref{math_peg} with the rule:
\[
 \NTexpr \ ::= \ \NTexpr \xx \term{{+}} \xx \NTfactor \choiceOPs \NTfactor
\]
Recursive descent parsing of expressions with this rule would start with 
recursively calling a function to parse expression on the same input, 
obviously leading to an infinite loop. But not only direct left recursion 
must be avoided. In the following rule:
\[
 \NTid{A} \ ::= \ \NTid{B} \choiceOPs \NTid{C} \xx \pnot{\:\!\NTid{D}} \xx \NTid{A}
\]
a similar problem occurs provided that $\NTid{B}$ may fail and $\NTid{C}$ 
and $\NTid{D}$ may succeed, the former without consuming any input.

While some techniques to deal with left-recursive PEGs have been developed 
recently \cite{WarEA08}, we choose to simply reject such grammars. 
In general it is undecidable whether a PEG grammar is complete, as it is 
undecidable whether the language generated by $\GG$ is empty \cite{For04}. 

While in general checking grammar completeness is undecidable, 
we follow Ford~\cite{For04} to develop
a simple syntactical check for \emph{well-formedness} of a grammar, which
implies its completeness. This check will reject left-recursive grammars
even if the part with left-recursion is unreachable in the grammar, 
but from a practical point of view this is hardly a limitation.

\subsection{PEG analysis}\label{sec:pegs-wf-props}

We define the \emph{expression set} of $\GG$ as:
\[
 \ES(\GG) = \set{e' \mid e' \subexp e, e \in \PD(A), A \in \NT}
\]
where $\subexp$ is a (non-strict) sub-expression relation on parsing expressions.

\begin{figure}[t!]
\begin{center}
\[
\begin{array}{ccccccc}
\inferrule{ }{\propE{\emptye}}   &
\quad \inferrule{ }{\propNE{\anychar}} &
\quad \inferrule{ }{\propF{\anychar}}  &
\quad \inferrule{a \in \T}{\propNE{\term{a}}} &
\quad \inferrule{a \in \T}{\propF{\term{a}}}   &
\quad \inferrule{\propF{e}}{\propE{\prep{e}}}   &
\quad \inferrule{\propNE{e}}{\propNE{\prep{e}}}
\end{array}
\]
\[
\begin{array}{c@{\qquad}c}
\inferrule{\star \in \set{\propEsym, \propNEsym, \propFsym} \\ A \in \NT \\ \pegProp{\PD(A)}{\star}}{\pegProp{A}{\star}} &
\inferrule{\propF{e_1} \lor (\propNF{e_1} \land \propF{e_2})}{\propF{\pseq{e_1}{e_2}}} \\[15pt]
\inferrule{(\propNE{e_1} \land \propNF{e_2}) \lor (\propNF{e_1} \land \propNE{e_2})}{\propNE{\pseq{e_1}{e_2}}} &
\inferrule{\propE{e_1} \\ \propE{e_2}}{\propE{\pseq{e_1}{e_2}}} \\[15pt]
\inferrule{\propE{e_1} \lor (\propF{e_1} \land \propE{e_2})}{\propE{\pchoice{e_1}{e_2}}} &
\inferrule{\propF{e_1} \\ \propF{e_2}}{\propF{\pchoice{e_1}{e_2}}} \\[15pt]
\inferrule{\propNE{e_1} \lor (\propF{e_1} \land \propNE{e_2})}{\propNE{\pchoice{e_1}{e_2}}} &
\inferrule{\propF{e}}{\propE{\pnot{e}}}   \qquad\qquad
\inferrule{\propNF{e}}{\propF{\pnot{e}}}  
\end{array}
\]
\end{center}
\caption{Deriving grammar properties.}
\label{peg-prop-fig}
\end{figure}

We define three groups of properties over parsing expressions: 
\begin{iteMize}{$\bullet$}
  \item ``$\propEsym$'': parsing expression can succeed without consuming any input,
  \item ``$\propNEsym$'': parsing expression can succeed after consuming some input and
  \item ``$\propFsym$'': parsing expression can fail.
\end{iteMize}

We will write $\propE{e}$ to indicate that the expression $e$ has property ``$\propEsym$''
(similarly for $\propSNE$ and $\propSF$). We will also write $\propNF{e}$ to denote
$\propE{e} \lor \propNE{e}$. We define inference rules for deriving those properties 
in Figure~\ref{peg-prop-fig}.

We start with empty sets of properties and apply those inference rules 
over $\ES(\GG)$ until reaching a fix-point. 
The existence of the fix-point is ensured by the fact that we extend those
property sets monotonically and they are bounded by the finite set $\ES(\GG)$. We 
summarize the semantics of those properties in the following lemma:

\begin{lem}[\cite{For04}]\label{th:pegProp}
For arbitrary $e \in \pex$ and $s \in \String$:
\begin{iteMize}{$\bullet$}
 \item if $\pegsemx{e}{s}{n}{\ok{s}}$ then $\propE{e}$,
 \item if $\pegsemx{e}{s}{n}{\ok{s'}}$ and $\strlen{s'} < \strlen{s}$ then $\propNE{e}$ and
 \item if $\pegsemx{e}{s}{n}{\fail}$ then $\propF{e}$.
\end{iteMize}
\end{lem}
\begin{proof}
Induction over $n$. All cases easy by the induction hypothesis
and semantical rules of XPEGs, except for $\prep{e}$ which requires
use of Lemma~\ref{thm:peg_loop_cond}.
\end{proof}

Those properties will be used for establishing well-formedness of a PEG, as we will see
in the following section. It is worth noting here that checking whether
$\propE{e}$ also plays a crucial role in the formal approach to parsing developed by
Danielsson~\cite{Dan10} (we will say more about his work in Section~\ref{sec:relwork}).

It is also interesting to consider such a simplified analysis in our setting, \ie,
only considering $\propE{e}$ and collapsing derivations of Figure~\ref{peg-prop-fig}
by assuming $\propNE{e}$ and $\propF{e}$ hold for every expression $e$. At first
it seems we would lose some precision by such an over-approximation as for instance
that would lead us to conclude $\propE{\pnot{\emptye}}$, whereas in fact this
expression can never succeed without consuming any input (as, quite simply, it
can \emph{never} succeed). As we will see soon this would lead us to reject
a valid definition:
\[
 \NTid{A} \ ::= \ \pnot{\emptye}\, ;\, \NTid{A}
\]
However, this definition of $\NTid{A}$ is not very interesting as it always fails.
In fact, we conjecture that the differences occur only in such degenerated cases
and that in practice such a simplified analysis would be as efficient as that
of~\cite{For04}.

\subsection{PEG well-formedness}\label{sec:pegs-wf-wf}

Using the semantics of those properties of parsing expression we can perform the 
completeness analysis of $\GG$. We introduce a set of well-formed expressions
$\WFset$ and again iterate from an empty set by using derivation rules from
Figure~\ref{peg-wf-fig} over $\ES(\GG)$ until reaching a fix-point. 

\begin{figure}[t!]
\begin{center}
\[
\begin{array}{c@{\quad}c@{\quad}c}
\inferrule{A \in \NT \\ \WF{\PD(A)}}{\WF{A}} &
\inferrule{ }{\WF{\emptye}} \qquad
\inferrule{ }{\WF{\anychar}} &
\inferrule{a \in \T}{\WF{\term{a}}} \qquad
\inferrule{\WF{e}}{\WF{\pnot{e}}} \\[15pt]
\inferrule{\WF{e_1} \\ \propE{e_1} \imp \WF{e_2}}{\WF{\pseq{e_1}{e_2}}} &
\inferrule{\WF{e_1} \\ \WF{e_2}}{\WF{\pchoice{e_1}{e_2}}} &
\inferrule{\WF{e},\quad\notPropE{e}}{\WF{\prep{e}}}
\end{array}
\]
\end{center}
\caption{Deriving the well-formedness property for a PEG.}
\label{peg-wf-fig}
\end{figure}

We say that $\GG$ is well-formed if $\ES(\GG) = \WFset$. We have the following result:

\newcommand{\IHin}{\ensuremath{\mathid{IH}_{\mathid{in}}}}
\newcommand{\IHout}{\ensuremath{\mathid{IH}_{\mathid{out}}}}
\begin{thm}[\cite{For04}]\label{th:pegs-wf-total}
If $\GG$ is well-formed then it is complete.
\end{thm}
\begin{proof}
We will say that $(e, s)$ is complete iff
\(
  \exists_{n, r}\; \pegsemx{e}{s}{n}{r}
\).
So we have to prove that $(e, s)$ is complete for all
$e \in \ES(\GG)$ and all strings $s$. We proceed
by induction over the length of the string $s$ (\IHout), 
followed by induction on the depth of the derivation tree of $\WF{e}$ 
(\IHin).
So we have to prove correctness of a one step derivation of the
well-formedness property (Figure~\ref{peg-wf-fig}) assuming
that all expressions are total on shorter strings. 
The interesting cases are:
\begin{iteMize}{$\bullet$}
  \item For a sequence $\pseq{e_1}{e_2}$ if $\WF{\pseq{e_1}{e_2}}$ 
    then $\WF{e_1}$, so $(e_1, s)$ is complete by \IHin. 
    If $e_1$ fails then $e_1; e_2$ fails.
    Otherwise $\pegsemx{e_1}{s}{n}{\Ok{v}{s'}}$. If $s = s'$ then
    $\propE{e_1}$ (Lemma~\ref{th:pegProp}) and hence 
    $\WF{e_2}$ and $(e_2, s')$ is complete by \IHin.
    If $s \neq s'$ then $\strlen{s'} < \strlen{s}$ 
    (Theorem~\ref{thm:peg_prefix}) and $(e_2, s')$ is complete by \IHout.
    Either way $(e_2, s')$ is complete and we conclude by semantical
    rules for sequence.
  \item For a repetition $\prep{e}$, $\WF{e}$ gives us completeness
    of $(e, s)$ by \IHin. If $e$ fails then we conclude 
    by the base rule for repetition. Otherwise $\pegsemx{\prep{e}}{s}{n}{s'}$ with
    $\strlen{s'} < \strlen{s}$ as $\notPropE{e}$. Hence we get
    completeness of $(\prep{e}, s')$ by \IHout\ and we conclude
    with the inductive rule for repetition. \qedhere
\end{iteMize}
\end{proof}

\section{Formally Verified XPEG interpreter}\label{sec:pegs-int}

In this Section we will present a \coq\ implementation of a 
parser interpreter. This task consists of formalizing the
theory of the preceding sections and, based on this,
writing an interpreter for well-formed XPEGs along with 
its correctness proofs.
The development is too big to present it in detail here, but 
we will try to comment on its most interesting aspects.

We will describe how PEGs are expressed in \coq\ in 
Section~\ref{sec:trx-spec-peg}, comment on the procedure
for checking their well-formedness in Section~\ref{sec:trx-wf}
and describe the formal development of an XPEG 
interpreter in Section~\ref{sec:trx-int}.

\subsection{Specifying XPEGs in \coq}\label{sec:trx-spec-peg}

XPEGs in \coq\ are a simple reflection of Definition~\ref{pegs-prod-peg}.
They are specified over a finite enumeration of non-terminals 
(corresponding to $\NT$) with their types ($\PT$):

\begin{hscode}\SaveRestoreHook
\column{B}{@{}>{\hspre}l<{\hspost}@{}}%
\column{3}{@{}>{\hspre}l<{\hspost}@{}}%
\column{E}{@{}>{\hspre}l<{\hspost}@{}}%
\>[3]{}\Conid{Parameter}\;prod\mathbin{:}\Conid{Enumeration}.{}\<[E]%
\\
\>[3]{}\Conid{Parameter}\;\Varid{prod\char95 type}\mathbin{:}prod\to \Conid{Type}.{}\<[E]%
\ColumnHook
\end{hscode}\resethooks
Building on that we define:
\begin{iteMize}{$\bullet$}
 \item \ensuremath{\Varid{pexp}}: un-typed parsing expressions, $\pex$, and
 \item \ensuremath{\Conid{PExp}}: their typed variant, $\Pex{\alpha}$, which follows the typing 
   discipline from Figure~\ref{peg-exp-prod-fig}.
\end{iteMize}
We present both definitions side by side:

\begin{minipage}[b]{0.3\linewidth}
\begin{hscode}\SaveRestoreHook
\column{B}{@{}>{\hspre}l<{\hspost}@{}}%
\column{E}{@{}>{\hspre}l<{\hspost}@{}}%
\>[B]{}\ensuremath{\mathbf{Inductive}}\;\Varid{pexp}\mathbin{:}\Conid{Type}\mathbin{:=}{}\<[E]%
\\
\>[B]{}\mid \Varid{empty}{}\<[E]%
\\
\>[B]{}\mid \Varid{anyChar}{}\<[E]%
\\
\>[B]{}\mid \Varid{terminal}\;(\Varid{a}\mathbin{:}\Varid{char}){}\<[E]%
\\
\>[B]{}\mid \Varid{range}\;(\Varid{a}\;\Varid{z}\mathbin{:}\Varid{char}){}\<[E]%
\\
\>[B]{}\mid \Varid{nonTerminal}\;(\Varid{p}\mathbin{:}prod){}\<[E]%
\\
\>[B]{}\mid \Varid{seq}\;(\Varid{e1}\;\Varid{e2}\mathbin{:}\Varid{pexp}){}\<[E]%
\\
\>[B]{}\mid \Varid{choice}\;(\Varid{e1}\;\Varid{e2}\mathbin{:}\Varid{pexp}){}\<[E]%
\\
\>[B]{}\mid \Varid{star}\;(\Varid{e}\mathbin{:}\Varid{pexp}){}\<[E]%
\\
\>[B]{}\mid not\;(\Varid{e}\mathbin{:}\Varid{pexp}){}\<[E]%
\\
\>[B]{}\mid \Varid{id}\;(\Varid{e}\mathbin{:}\Varid{pexp}).{}\<[E]%
\ColumnHook
\end{hscode}\resethooks
\end{minipage}
\hspace{0.5cm}
\begin{minipage}[b]{0.7\linewidth}
\begin{hscode}\SaveRestoreHook
\column{B}{@{}>{\hspre}l<{\hspost}@{}}%
\column{E}{@{}>{\hspre}l<{\hspost}@{}}%
\>[B]{}\ensuremath{\mathbf{Inductive}}\;\Conid{PExp}\mathbin{:}\Conid{Type}\to \Conid{Type}\mathbin{:=}{}\<[E]%
\\
\>[B]{}\mid \Conid{Empty}\mathbin{:}\Conid{PExp}\;\Conid{True}{}\<[E]%
\\
\>[B]{}\mid \Conid{AnyChar}\mathbin{:}\Conid{PExp}\;\Varid{char}{}\<[E]%
\\
\>[B]{}\mid \Conid{Terminal}\mathbin{:}\Varid{char}\to \Conid{PExp}\;\Varid{char}{}\<[E]%
\\
\>[B]{}\mid \Conid{Range}\mathbin{:}\Varid{char}\mathbin{*}\Varid{char}\to \Conid{PExp}\;\Varid{char}{}\<[E]%
\\
\>[B]{}\mid \Conid{NonTerminal}\mathbin{:}\ensuremath{\forall}\;\Varid{p},\Conid{PExp}\;(\Varid{prod\char95 type}\;\Varid{p}){}\<[E]%
\\
\>[B]{}\mid \Conid{Seq}\mathbin{:}\ensuremath{\forall}\;\Conid{A}\;\Conid{B},\Conid{PExp}\;\Conid{A}\to \Conid{PExp}\;\Conid{B}\to \Conid{PExp}\;(\Conid{A}\mathbin{*}\Conid{B}){}\<[E]%
\\
\>[B]{}\mid \Conid{Choice}\mathbin{:}\ensuremath{\forall}\;\Conid{A},\Conid{PExp}\;\Conid{A}\to \Conid{PExp}\;\Conid{A}\to \Conid{PExp}\;\Conid{A}{}\<[E]%
\\
\>[B]{}\mid \Conid{Star}\mathbin{:}\ensuremath{\forall}\;\Conid{A},\Conid{PExp}\;\Conid{A}\to \Conid{PExp}\;(\Varid{list}\;\Conid{A}){}\<[E]%
\\
\>[B]{}\mid \Conid{Not}\mathbin{:}\ensuremath{\forall}\;\Conid{A},\Conid{PExp}\;\Conid{A}\to \Conid{PExp}\;\Conid{True}{}\<[E]%
\\
\>[B]{}\mid \Conid{Action}\mathbin{:}\ensuremath{\forall}\;\Conid{A}\;\Conid{B},\Conid{PExp}\;\Conid{A}\to (\Conid{A}\to \Conid{B})\to \Conid{PExp}\;\Conid{B}.{}\<[E]%
\ColumnHook
\end{hscode}\resethooks
\end{minipage}

\noindent Those definitions are straight-forward encodings of
Definitions~\ref{pexp} and \ref{PExp}. We implemented the range
operator $\prange{a}{z}$ as a primitive, as in practice it occurs
frequently in parsers and implementing it as a derived operation by a
choice over all the characters in the range is inefficient. That means
that in the formalization we had to extend the semantics of
Figure~\ref{peg-sem-prod-fig} with this operator, in a straightforward
way.

It is worth noting here that \ensuremath{\Conid{PExp}} is \emph{large}, in terms of Coq universe
levels, as its index lives in \ensuremath{\Conid{Type}}. We never work with propositional equality
of types, so the constraints on types used in constructors of \ensuremath{\Conid{PExp}}, come only
from the inductive definition itself. In particular, \ensuremath{\Conid{PExp}} must live at
a higher universe level than any type used in its constructors.

For ``regular'' use of our parsing machinery this should pose no problems.
However, should we want to develop some higher-order grammars (grammars that
upon parsing return another grammar) we would very soon run into Coq's
\ensuremath{\Conid{Universe}\;\Conid{Inconsistency}} problems. In fact higher-order grammars are
not expressible in our framework anyway, due to the use of Coq's module
system. We will return to this issue in Section~\ref{sec:discussion}.

With \ensuremath{\Varid{pexp}} and \ensuremath{\Conid{PExp}} in place we continue by defining, in an obvious way, conversion
functions from one structure to the another.
\begin{hscode}\SaveRestoreHook
\column{B}{@{}>{\hspre}l<{\hspost}@{}}%
\column{3}{@{}>{\hspre}l<{\hspost}@{}}%
\column{E}{@{}>{\hspre}l<{\hspost}@{}}%
\>[3]{}\ensuremath{\mathbf{Fixpoint}}\;\Varid{pexp\char95 project}\;\Conid{T}\;(\Varid{e}\mathbin{:}\Conid{PExp}\;\Conid{T})\mathbin{:}\Varid{pexp}\mathbin{:=}\hspace{0.9mm}\{\mathbin{...}\mskip1.5mu\}{}\<[E]%
\\
\>[3]{}\ensuremath{\mathbf{Fixpoint}}\;\Varid{pexp\char95 promote}\;(\Varid{e}\mathbin{:}\Varid{pexp})\mathbin{:}\Conid{PExp}\;\Conid{True}\mathbin{:=}\hspace{0.9mm}\{\mathbin{...}\mskip1.5mu\}{}\<[E]%
\ColumnHook
\end{hscode}\resethooks
Conversion from \ensuremath{\Conid{PExp}} to \ensuremath{\Varid{pexp}} simply erases types and maps 
\ensuremath{\Conid{Action}}s to dummy constructor \ensuremath{\Varid{id}}. Conversion in the other
direction maps to expressions of a singleton type \ensuremath{\Conid{True}},
inserting, where needed, type coercions using \ensuremath{\Conid{Action}} operator.

To complete the definition of XPEG grammar, Definition~\ref{pegs-prod-peg},
we declare definitions of non-terminals ($\PD$) and the starting 
production ($\PS$) as:\begin{hscode}\SaveRestoreHook
\column{B}{@{}>{\hspre}l<{\hspost}@{}}%
\column{3}{@{}>{\hspre}l<{\hspost}@{}}%
\column{E}{@{}>{\hspre}l<{\hspost}@{}}%
\>[3]{}\Conid{Parameter}\;\Varid{production}\mathbin{:}\ensuremath{\forall}\;\Varid{p}\mathbin{:}prod,\Conid{PExp}\;(\Varid{prod\char95 type}\;\Varid{p}).{}\<[E]%
\\
\>[3]{}\Conid{Parameter}\;\Varid{start}\mathbin{:}prod.{}\<[E]%
\ColumnHook
\end{hscode}\resethooks

There are two observations that we would like to make at this point.
First, by means of the above embedding of XPEGs in \coq,
every such XPEG is well-defined (though not necessarily well-formed).
In particular there can be no calls to undefined non-terminals and the
conformance with the typing discipline from Figure~\ref{peg-exp-prod-fig} 
is taken care of by the type-checker of \coq.

Secondly, thanks to the use of \coq's mechanisms, such as notations 
and coercions, expressing an XPEG in \coq\ is still relatively easy
as we will see in the following example.

\begin{exa}\label{math_coq}
Figure~\ref{math_coq_peg} presents a precise \coq\ rendering of the 
productions of the XPEG grammar from Example~\ref{math_xpeg}. It is
not much more verbose than the original example. 
Each \ensuremath{\Conid{Pi}} function corresponds to $i$'th projection and they
work with arbitrary $n$-tuples thanks to the type-class mechanism.
\exqed
\end{exa}

\begin{figure}[t!]
\begin{center}
\begin{hscode}\SaveRestoreHook
\column{B}{@{}>{\hspre}l<{\hspost}@{}}%
\column{3}{@{}>{\hspre}l<{\hspost}@{}}%
\column{12}{@{}>{\hspre}l<{\hspost}@{}}%
\column{13}{@{}>{\hspre}l<{\hspost}@{}}%
\column{44}{@{}>{\hspre}l<{\hspost}@{}}%
\column{E}{@{}>{\hspre}l<{\hspost}@{}}%
\>[B]{}\ensuremath{\mathbf{Program}}\;\ensuremath{\mathbf{Definition}}\;\Varid{production}\;\Varid{p}\mathbin{:=}{}\<[E]%
\\
\>[B]{}\hsindent{3}{}\<[3]%
\>[3]{}\ensuremath{\mathbf{match}}\;\Varid{p}\;\ensuremath{\mathbf{return}}\;\Conid{PExp}\;(\Varid{prod\char95 type}\;\Varid{p})\;\ensuremath{\mathbf{with}}{}\<[E]%
\\
\>[B]{}\hsindent{3}{}\<[3]%
\>[3]{}\mid \Varid{ws}{}\<[12]%
\>[12]{}\Rightarrow (\text{\tt \char34 ~\char34}\mathbin{/}\text{\tt \char34 \char92 t\char34})\;[\mskip1.5mu \mathbin{*}\mskip1.5mu]\;{}\<[44]%
\>[44]{}[\mskip1.5mu \mathbin{\#}\mskip1.5mu]{}\<[E]%
\\
\>[B]{}\hsindent{3}{}\<[3]%
\>[3]{}\mid \Varid{number}\Rightarrow [\mskip1.5mu \text{\tt \char34 0\char34}\;\!\!-\!-\text{\tt \char34 9\char34}\mskip1.5mu]\;[\mskip1.5mu \mathbin{+}\mskip1.5mu]\;{}\<[44]%
\>[44]{}[\mskip1.5mu \to \mskip1.5mu]\;\Varid{digListToRat}{}\<[E]%
\\
\>[B]{}\hsindent{3}{}\<[3]%
\>[3]{}\mid \Varid{term}{}\<[12]%
\>[12]{}\Rightarrow \Varid{ws};\Varid{number};\Varid{ws}\;{}\<[44]%
\>[44]{}[\mskip1.5mu \to \mskip1.5mu]\;(\ensuremath{\lambda\hspace{-1mm} }\;\Varid{v}\Rightarrow \Conid{P2}\;\Varid{v}){}\<[E]%
\\
\>[12]{}\hsindent{1}{}\<[13]%
\>[13]{}\mathbin{/}\Varid{ws};\text{\tt \char34 (\char34};\Varid{expr};\text{\tt \char34 )\char34};\Varid{ws}\;{}\<[44]%
\>[44]{}[\mskip1.5mu \to \mskip1.5mu]\;(\ensuremath{\lambda\hspace{-1mm} }\;\Varid{v}\Rightarrow \Conid{P3}\;\Varid{v}){}\<[E]%
\\
\>[B]{}\hsindent{3}{}\<[3]%
\>[3]{}\mid \Varid{factor}\Rightarrow \Varid{term};\text{\tt \char34 *\char34};\Varid{factor}\;{}\<[44]%
\>[44]{}[\mskip1.5mu \to \mskip1.5mu]\;(\ensuremath{\lambda\hspace{-1mm} }\;\Varid{v}\Rightarrow \Conid{P1}\;\Varid{v}\mathbin{*}\Conid{P3}\;\Varid{v}){}\<[E]%
\\
\>[3]{}\hsindent{10}{}\<[13]%
\>[13]{}\mathbin{/}\Varid{term}{}\<[E]%
\\
\>[B]{}\hsindent{3}{}\<[3]%
\>[3]{}\mid \Varid{expr}{}\<[12]%
\>[12]{}\Rightarrow \Varid{factor};\text{\tt \char34 +\char34};\Varid{expr}\;{}\<[44]%
\>[44]{}[\mskip1.5mu \to \mskip1.5mu]\;(\ensuremath{\lambda\hspace{-1mm} }\;\Varid{v}\Rightarrow \Conid{P1}\;\Varid{v}\mathbin{+}\Conid{P3}\;\Varid{v}){}\<[E]%
\\
\>[12]{}\hsindent{1}{}\<[13]%
\>[13]{}\mathbin{/}\Varid{factor}{}\<[E]%
\\
\>[B]{}\hsindent{3}{}\<[3]%
\>[3]{}\ensuremath{\mathbf{end}}.{}\<[E]%
\ColumnHook
\end{hscode}\resethooks
\end{center}
\caption{A \coq\ version of the XPEG for mathematical expressions from Example~\ref{math_xpeg}}
\label{math_coq_peg}
\end{figure}

\subsection{Checking well-formedness of an XPEG}\label{sec:trx-wf}

To check well-formedness of XPEGs we implement the procedure from 
Section~\ref{sec:pegs-wf}. It is worth noting that the
function to compute XPEG properties, by iterating the derivation 
rules of Figure~\ref{peg-prop-fig} until reaching a fix-point, is not structurally 
recursive. Similarly for the well-formedness check with rules from
Figure~\ref{peg-wf-fig}. Fortunately the Program feature~\cite{Soz07} of \coq\
makes specifying such functions much easier. We illustrate it on
the well-formedness check (computing properties is analogous).

We begin by one-step well-formedness derivation corresponding
to Figure~\ref{peg-wf-fig}.\begin{hscode}\SaveRestoreHook
\column{B}{@{}>{\hspre}l<{\hspost}@{}}%
\column{3}{@{}>{\hspre}l<{\hspost}@{}}%
\column{5}{@{}>{\hspre}l<{\hspost}@{}}%
\column{E}{@{}>{\hspre}l<{\hspost}@{}}%
\>[3]{}\ensuremath{\mathbf{Definition}}\;\Varid{wf\char95 analyse}\;(\Varid{exp}\mathbin{:}\Varid{pexp})\;(\Varid{wf}\mathbin{:}\Varid{\Conid{PES}.t})\mathbin{:}\Varid{bool}\mathbin{:=}{}\<[E]%
\\
\>[3]{}\hsindent{2}{}\<[5]%
\>[5]{}\ensuremath{\mathbf{match}}\;\Varid{exp}\;\ensuremath{\mathbf{with}}{}\<[E]%
\\
\>[3]{}\hsindent{2}{}\<[5]%
\>[5]{}\mid \Varid{empty}\Rightarrow \Varid{true}{}\<[E]%
\\
\>[3]{}\hsindent{2}{}\<[5]%
\>[5]{}\mid \Varid{range}\;\anonymous \;\anonymous \Rightarrow \Varid{true}{}\<[E]%
\\
\>[3]{}\hsindent{2}{}\<[5]%
\>[5]{}\mid \Varid{terminal}\;\Varid{a}\Rightarrow \Varid{true}{}\<[E]%
\\
\>[3]{}\hsindent{2}{}\<[5]%
\>[5]{}\mid \Varid{anyChar}\Rightarrow \Varid{true}{}\<[E]%
\\
\>[3]{}\hsindent{2}{}\<[5]%
\>[5]{}\mid \Varid{nonTerminal}\;\Varid{p}\Rightarrow \Varid{is\char95 wf}\;(\Varid{production}\;\Varid{p})\;\Varid{wf}{}\<[E]%
\\
\>[3]{}\hsindent{2}{}\<[5]%
\>[5]{}\mid \Varid{seq}\;\Varid{e1}\;\Varid{e2}\Rightarrow \Varid{is\char95 wf}\;\Varid{e1}\;\Varid{wf}\mathrel{\wedge}(\mathbf{if}\;\Varid{e1}\mathbin{-}[\mskip1.5mu \Varid{gp}\mskip1.5mu]\to \mathrm{0}\;\mathbf{then}\;\Varid{is\char95 wf}\;\Varid{e2}\;\Varid{wf}\;\mathbf{else}\;\Varid{true}){}\<[E]%
\\
\>[3]{}\hsindent{2}{}\<[5]%
\>[5]{}\mid \Varid{choice}\;\Varid{e1}\;\Varid{e2}\Rightarrow \Varid{is\char95 wf}\;\Varid{e1}\;\Varid{wf}\mathrel{\wedge}\Varid{is\char95 wf}\;\Varid{e2}\;\Varid{wf}{}\<[E]%
\\
\>[3]{}\hsindent{2}{}\<[5]%
\>[5]{}\mid \Varid{star}\;\Varid{e}\Rightarrow \Varid{is\char95 wf}\;\Varid{e}\;\Varid{wf}\mathrel{\wedge}(\Varid{negb}\;(\Varid{e}\mathbin{-}[\mskip1.5mu \Varid{gp}\mskip1.5mu]\to \mathrm{0})){}\<[E]%
\\
\>[3]{}\hsindent{2}{}\<[5]%
\>[5]{}\mid not\;\Varid{e}\Rightarrow \Varid{is\char95 wf}\;\Varid{e}\;\Varid{wf}{}\<[E]%
\\
\>[3]{}\hsindent{2}{}\<[5]%
\>[5]{}\mid \Varid{id}\;\Varid{e}\Rightarrow \Varid{is\char95 wf}\;\Varid{e}\;\Varid{wf}{}\<[E]%
\\
\>[3]{}\hsindent{2}{}\<[5]%
\>[5]{}\ensuremath{\mathbf{end}}.{}\<[E]%
\ColumnHook
\end{hscode}\resethooks
This function take a set of well-formed expressions computed so far (\ensuremath{\Conid{PES}}
standing for ``parsing expression set'') and an expression \ensuremath{\Varid{exp}} and
returns true iff \ensuremath{\Varid{exp}} should also be consider well-formed, according
to the derivation system of Figure~\ref{peg-wf-fig}. Here \ensuremath{\Varid{gp}} is the
set of global properties computed following the procedure of
Section~\ref{sec:pegs-wf-props} (again, we do not show the code here,
as that procedure is very analogous to the inference of well-formedness,
that we describe). Hence \ensuremath{\Varid{e}\mathbin{-}[\mskip1.5mu \Varid{gp}\mskip1.5mu]\to \mathrm{0}} should be read as $\propE{e}$ and
\ensuremath{\Varid{is\char95 wf}} is an abbreviation for set membership, \ie:\begin{hscode}\SaveRestoreHook
\column{B}{@{}>{\hspre}l<{\hspost}@{}}%
\column{3}{@{}>{\hspre}l<{\hspost}@{}}%
\column{E}{@{}>{\hspre}l<{\hspost}@{}}%
\>[3]{}\ensuremath{\mathbf{Definition}}\;\Varid{is\char95 wf}\mathbin{:}\Varid{pexp}\to \Varid{\Conid{PES}.t}\to \Varid{bool}\mathbin{:=}\Varid{\Conid{PES}.mem}.{}\<[E]%
\ColumnHook
\end{hscode}\resethooks

With that in place we continue with a simple function that extends the
set of well-formed expressions with the one being considered now, in
case it was established to be well-formed by invocation of \ensuremath{\Varid{wf\char95 analyse}}
and otherwise leaves this set unchanged.\begin{hscode}\SaveRestoreHook
\column{B}{@{}>{\hspre}l<{\hspost}@{}}%
\column{3}{@{}>{\hspre}l<{\hspost}@{}}%
\column{5}{@{}>{\hspre}l<{\hspost}@{}}%
\column{E}{@{}>{\hspre}l<{\hspost}@{}}%
\>[3]{}\ensuremath{\mathbf{Definition}}\;\Varid{wf\char95 analyse\char95 exp}\;(\Varid{exp}\mathbin{:}\Varid{pexp})\;(\Varid{wf}\mathbin{:}\Varid{\Conid{PES}.t})\mathbin{:}\Varid{\Conid{PES}.t}\mathbin{:=}{}\<[E]%
\\
\>[3]{}\hsindent{2}{}\<[5]%
\>[5]{}\mathbf{if}\;\Varid{wf\char95 analyse}\;\Varid{exp}\;\Varid{wf}\;\mathbf{then}\;\Varid{\Conid{PES}.add}\;\Varid{exp}\;\Varid{wf}\;\mathbf{else}\;\Varid{wf}.{}\<[E]%
\ColumnHook
\end{hscode}\resethooks

Now the one step derivation over all expressions $\ES(\GG)$, represented
by the constant \ensuremath{\Varid{grammarExpSet}} below, can be realized as a simple fold
operation using the above function:\begin{hscode}\SaveRestoreHook
\column{B}{@{}>{\hspre}l<{\hspost}@{}}%
\column{3}{@{}>{\hspre}l<{\hspost}@{}}%
\column{5}{@{}>{\hspre}l<{\hspost}@{}}%
\column{E}{@{}>{\hspre}l<{\hspost}@{}}%
\>[3]{}\ensuremath{\mathbf{Definition}}\;\Varid{wf\char95 derive}\;(\Varid{wf}\mathbin{:}\Varid{\Conid{PES}.t})\mathbin{:}\Varid{\Conid{PES}.t}\mathbin{:=}{}\<[E]%
\\
\>[3]{}\hsindent{2}{}\<[5]%
\>[5]{}\Varid{\Conid{PES}.fold}\;\Varid{wf\char95 analyse\char95 exp}\;\Varid{grammarExpSet}\;\Varid{wf}.{}\<[E]%
\ColumnHook
\end{hscode}\resethooks

Now, the complete analysis is a fixpoint of applying one-step derivation
\ensuremath{\Varid{wf\char95 derive}}.
\begin{hscode}\SaveRestoreHook
\column{B}{@{}>{\hspre}l<{\hspost}@{}}%
\column{3}{@{}>{\hspre}l<{\hspost}@{}}%
\column{E}{@{}>{\hspre}l<{\hspost}@{}}%
\>[B]{}\ensuremath{\mathbf{Program}}\;\ensuremath{\mathbf{Fixpoint}}\;\Varid{wf\char95 compute}\;(\Varid{wf}\mathbin{:}\Conid{WFset})\;\hspace{0.9mm}\{\ensuremath{\mathbf{measure}}\;\Varid{wf\char95 measure}\;\Varid{wf}\mskip1.5mu\}\mathbin{:}\Conid{WFset}\mathbin{:=}{}\<[E]%
\\
\>[B]{}\hsindent{3}{}\<[3]%
\>[3]{}\mathbf{let}\;\Varid{wf'}\mathbin{:=}\Varid{wf\char95 derive}\;\Varid{wf}\;\mathbf{in}{}\<[E]%
\\
\>[B]{}\hsindent{3}{}\<[3]%
\>[3]{}\mathbf{if}\;\Varid{\Conid{PES}.equal}\;\Varid{wf}\;\Varid{wf'}\;\mathbf{then}\;\Varid{wf}\;\mathbf{else}\;\Varid{wf\char95 compute}\;\Varid{wf'}.{}\<[E]%
\ColumnHook
\end{hscode}\resethooks
Here \ensuremath{\Conid{WFset}} is a set of well-formed expressions:\begin{hscode}\SaveRestoreHook
\column{B}{@{}>{\hspre}l<{\hspost}@{}}%
\column{3}{@{}>{\hspre}l<{\hspost}@{}}%
\column{E}{@{}>{\hspre}l<{\hspost}@{}}%
\>[3]{}\ensuremath{\mathbf{Definition}}\;\Conid{WFset}\mathbin{:=}\hspace{0.9mm}\{\Varid{e}\mathbin{:}\Varid{\Conid{PES}.t}\mid \Varid{wf\char95 prop}\;\Varid{e}\mskip1.5mu\}{}\<[E]%
\ColumnHook
\end{hscode}\resethooks
where \ensuremath{\Varid{wf\char95 prop}} is a predicate capturing well-formedness of an expression.

The main difficulty here is that \ensuremath{\Varid{wf\char95 compute}} is not structurally recursive.
However, we can construct a measure (into $\nat$) that will decrease
along recursive calls as:
\[
  \ensuremath{\Varid{wf\char95 measure}} ::= \strlen{\ES(\GG)} - \strlen{\ensuremath{\Varid{wf}}}
\]
Now we can prove this procedure terminating, as the set of well-formed
expressions is growing monotonically and is contained in $\ES(\GG)$:
\begin{gather*}
  \ensuremath{\Varid{wf}} \subseteq \ensuremath{\Varid{wf\char95 derive}\;\Varid{wf}} \\
  \ensuremath{\Varid{wf}} \subseteq \ES(\GG) \implies \ensuremath{\Varid{wf\char95 derive}\;\Varid{wf}} \subseteq \ES(\GG) 
\end{gather*}\smallskip

\noindent The Program feature~\cite{Soz07} of \coq, is very helpful in
expressing such non structurally recursive functions, as well as in
general programming with dependent types. The downside of Program is
that it inserts type casts, making reasoning about such functions more
difficult. This can be usually overcome with the use of sigma-types
capturing the function specification (\ensuremath{\Varid{wf\char95
    prop}} in our example) together with its return value. This style
of programming seems to be particularly well suited when working with
Program.

  Finally we obtain the set of well-formed expressions of a grammar by
iterating to a fix-point, starting with an empty set:\begin{hscode}\SaveRestoreHook
\column{B}{@{}>{\hspre}l<{\hspost}@{}}%
\column{3}{@{}>{\hspre}l<{\hspost}@{}}%
\column{E}{@{}>{\hspre}l<{\hspost}@{}}%
\>[3]{}\ensuremath{\mathbf{Program}}\;\ensuremath{\mathbf{Definition}}\;\Conid{WFexps}\mathbin{:}\Varid{\Conid{PES}.t}\mathbin{:=}\Varid{wf\char95 compute}\;\Varid{\Conid{PES}.empty}.{}\<[E]%
\ColumnHook
\end{hscode}\resethooks
a grammar expression \ensuremath{\Varid{exp}} is well-formed if it belongs to this set\begin{hscode}\SaveRestoreHook
\column{B}{@{}>{\hspre}l<{\hspost}@{}}%
\column{3}{@{}>{\hspre}l<{\hspost}@{}}%
\column{E}{@{}>{\hspre}l<{\hspost}@{}}%
\>[3]{}\ensuremath{\mathbf{Definition}}\;\Conid{WF}\;(\Varid{exp}\mathbin{:}\Varid{pexp})\mathbin{:}\Conid{Prop}\mathbin{:=}\Conid{\Conid{PES}.In}\;\Varid{exp}\;\Conid{WFexps}.{}\<[E]%
\ColumnHook
\end{hscode}\resethooks
and a grammar is well-formed if all its expressions are well-formed:\begin{hscode}\SaveRestoreHook
\column{B}{@{}>{\hspre}l<{\hspost}@{}}%
\column{3}{@{}>{\hspre}l<{\hspost}@{}}%
\column{E}{@{}>{\hspre}l<{\hspost}@{}}%
\>[3]{}\ensuremath{\mathbf{Definition}}\;\Varid{grammar\char95 WF}\mathbin{:}\Conid{Prop}\mathbin{:=}\Varid{grammarExpSet}\;[\mskip1.5mu \mathrel{=}\mskip1.5mu]\;\Conid{WFexps}.{}\<[E]%
\ColumnHook
\end{hscode}\resethooks

Above we presented a complete code of the well-formedness analysis
(Section~\ref{sec:pegs-wf-wf}), excluding the inference of properties
(Section~\ref{sec:pegs-wf-props}). Naturally, every of those functions
is accompanied with some lemmas stating its correctness and their
proofs. Those proofs, with Ltac definitions used to discard them,
constitute roughly 4-5x the size of the definitions. This factor
is so low thanks to heavy use of Ltac automation in the proofs;
the proof style advocated by Chlipala~\cite{Chl09}, which we, eventually,
learned to embrace fully.

\bigskip
Our interpreter (more on it in the following section) will work on
XPEGs, not on PEGs. However, the termination analysis sketched above
considers un-typed parsing expressions \ensuremath{\Varid{pexp}}, obtained by projecting
XPEGs expressions (with \ensuremath{\Varid{pexp\char95 project}}). The reason is two-fold. 

Firstly, semantic actions are embedded in Coq's programming language
and hence are terminating and have no influence on the termination
analysis of the grammar. Hence a termination of the parser on
expression \ensuremath{\Varid{e}\mathbin{:}\Conid{PExp}\;\Conid{T}} is immediate from termination of 
\ensuremath{\Varid{pexp\char95 project}\;\Varid{e}\mathbin{:}\Varid{pexp}}.

Secondly, the well-formedness procedure presented above 
needs to maintain a set of parsing expressions (\ensuremath{\Conid{WFset}}) and for that
we need a decidable equality over parsing expressions. Equality over
$\Pex{\alpha}$ is not decidable, as, within coercion operator
$\pcoerce{e}{f}$ they contain arbitrary functions $f$.

An alternative approach would be to consider \ensuremath{\Conid{WFset}} modulo
an equivalence relation on parsing expressions coarser than
the syntactic equality, which would ignore $f$ components
in $\pcoerce{e}{f}$ coercions. That would avoid formalization
of the un-typed structure \ensuremath{\Varid{pexp}} altogether for the price
of reasoning with dependently typed \ensuremath{\Conid{PExp}}'s in the well-formedness
analysis. 

\subsection{A formal interpreter for XPEGs}\label{sec:trx-int}

For the development of a formal interpreter for XPEGs 
we used the \ensuremath{\Varid{ascii}} type of \coq\ for the set of terminals $\T$.
The string type from the standard library of \coq\ is
isomorphic to lists of characters. In its place we just used
a list of characters, in order to be able to re-use a rich set
of available functions over lists.

First let us define the result of parsing an expression \ensuremath{\Conid{PExp}\;\Conid{T}}
on some string:\begin{hscode}\SaveRestoreHook
\column{B}{@{}>{\hspre}l<{\hspost}@{}}%
\column{3}{@{}>{\hspre}l<{\hspost}@{}}%
\column{E}{@{}>{\hspre}l<{\hspost}@{}}%
\>[3]{}\ensuremath{\mathbf{Inductive}}\;\Conid{ParsingResult}\;(\Conid{T}\mathbin{:}\Conid{Type})\mathbin{:}\Conid{Type}\mathbin{:=}{}\<[E]%
\\
\>[3]{}\mid \Conid{PR\char95 fail}.{}\<[E]%
\\
\>[3]{}\mid \Conid{PR\char95 ok}\;(\Varid{s}\mathbin{:}\Varid{string})\;(\Varid{v}\mathbin{:}\Conid{T}){}\<[E]%
\ColumnHook
\end{hscode}\resethooks
\ie, a parsing can either fail (\ensuremath{\Conid{PR\char95 fail}}) or succeed (\ensuremath{\Conid{PR\char95 ok}\;\Varid{s}\;\Varid{v}}),
in which case we obtain a suffix \ensuremath{\Varid{s}} that remains to be parsed and
an associated semantic value \ensuremath{\Varid{v}}.

Now after requiring a well-formed grammar, interpreter can be defined
as a function with the following header:\begin{hscode}\SaveRestoreHook
\column{B}{@{}>{\hspre}l<{\hspost}@{}}%
\column{3}{@{}>{\hspre}l<{\hspost}@{}}%
\column{4}{@{}>{\hspre}l<{\hspost}@{}}%
\column{E}{@{}>{\hspre}l<{\hspost}@{}}%
\>[3]{}\ensuremath{\mathbf{Variable}}\;\Conid{GWF}\mathbin{:}\Varid{grammar\char95 WF}.{}\<[E]%
\\
\>[B]{}\ensuremath{\mathbf{Program}}\;\ensuremath{\mathbf{Fixpoint}}\;\Varid{parse}\;(\Conid{T}\mathbin{:}\Conid{Type})\;(\Varid{e}\mathbin{:}\Conid{PExp}\;\Conid{T}\mid \Varid{is\char95 grammar\char95 exp}\;\Varid{e})\;(\Varid{s}\mathbin{:}\Varid{string})\;{}\<[E]%
\\
\>[B]{}\hsindent{4}{}\<[4]%
\>[4]{}\hspace{0.9mm}\{\ensuremath{\mathbf{measure}}\;(\Varid{e},\Varid{s})\succ\mskip1.5mu\}\mathbin{:}\hspace{0.9mm}\{\Varid{r}\mathbin{:}\Conid{ParsingResult}\;\Conid{T}\mid \ensuremath{\exists }\;\Varid{n},[\mskip1.5mu \Varid{e},\Varid{s}\mskip1.5mu]\Rightarrow [\mskip1.5mu \Varid{n},\Varid{r}\mskip1.5mu]\mskip1.5mu\}{}\<[E]%
\ColumnHook
\end{hscode}\resethooks
So this function takes three arguments (the first one implicit):
\begin{iteMize}{$\bullet$}
 \item \ensuremath{\Conid{T}}: a type of the result of parsing ($\alpha$),
 \item \ensuremath{\Varid{e}}: a parsing expression of type \ensuremath{\Conid{T}} ($\Pex{\alpha}$), 
   with a proof (\ensuremath{\Varid{is\char95 grammar\char95 exp}\;\Varid{e}}) that it belongs to the grammar 
   $\GG$ (which in turn is checked beforehand to be well-formed) and
 \item \ensuremath{\Varid{s}}: a string to be parsed.
\end{iteMize}

The last line in the above header describes the type of the
result of this function, where \ensuremath{[\mskip1.5mu \Varid{e},\Varid{s}\mskip1.5mu]\Rightarrow [\mskip1.5mu \Varid{n},\Varid{r}\mskip1.5mu]} is the expected 
encoding of the semantics from Figure~\ref{peg-sem-prod-fig} and 
corresponds to $\pegsemx{e}{s}{n}{r}$. So the \ensuremath{\Varid{parse}} function produces 
the parsing result $r$ (either $\fail$ or $\Ok{v}{s}$, 
with $\typei{v}{T}$), such that $\pegsemx{e}{s}{n}{r}$ for some $n$,
\ie, it is correct with respect to the semantic of XPEGs.

The body of the \ensuremath{\Varid{parse}} function performs pattern matching on expression 
\ensuremath{\Varid{e}} and interprets it according to the semantics from Figure~\ref{peg-sem-fig}.
We show a simplified (the actual pattern matching is slightly more involved due to
dealing with dependent types) excerpt of this function for a few types of expressions:\begin{hscode}\SaveRestoreHook
\column{B}{@{}>{\hspre}l<{\hspost}@{}}%
\column{4}{@{}>{\hspre}l<{\hspost}@{}}%
\column{8}{@{}>{\hspre}l<{\hspost}@{}}%
\column{10}{@{}>{\hspre}l<{\hspost}@{}}%
\column{16}{@{}>{\hspre}l<{\hspost}@{}}%
\column{E}{@{}>{\hspre}l<{\hspost}@{}}%
\>[4]{}\ensuremath{\mathbf{match}}\;\Varid{e}\;\ensuremath{\mathbf{with}}{}\<[E]%
\\
\>[4]{}\mid \Conid{Empty}\Rightarrow \Conid{Ok}\;\Varid{s}\;\Conid{I}{}\<[E]%
\\
\>[4]{}\mid \Conid{Terminal}\;\Varid{c}\Rightarrow {}\<[E]%
\\
\>[4]{}\hsindent{6}{}\<[10]%
\>[10]{}\ensuremath{\mathbf{match}}\;\Varid{s}\;\ensuremath{\mathbf{with}}{}\<[E]%
\\
\>[4]{}\hsindent{6}{}\<[10]%
\>[10]{}\mid \Varid{nil}\Rightarrow \Conid{Fail}{}\<[E]%
\\
\>[4]{}\hsindent{6}{}\<[10]%
\>[10]{}\mid \Varid{x}\mathbin{::}\Varid{xs}\Rightarrow {}\<[E]%
\\
\>[10]{}\hsindent{6}{}\<[16]%
\>[16]{}\ensuremath{\mathbf{match}}\;\Varid{\Conid{CharAscii}.eq\char95 dec}\;\Varid{c}\;\Varid{x}\;\ensuremath{\mathbf{with}}{}\<[E]%
\\
\>[10]{}\hsindent{6}{}\<[16]%
\>[16]{}\mid \Varid{left}\;\anonymous \Rightarrow \Conid{Ok}\;\Varid{xs}\;\Varid{c}{}\<[E]%
\\
\>[10]{}\hsindent{6}{}\<[16]%
\>[16]{}\mid \Varid{right}\;\anonymous \Rightarrow \Conid{Fail}{}\<[E]%
\\
\>[10]{}\hsindent{6}{}\<[16]%
\>[16]{}\ensuremath{\mathbf{end}}{}\<[E]%
\\
\>[4]{}\hsindent{4}{}\<[8]%
\>[8]{}\ensuremath{\mathbf{end}}{}\<[E]%
\\
\>[4]{}\mid \Conid{NonTerminal}\;\Varid{p}\Rightarrow \Varid{parse}\;(\Varid{production}\;\Varid{p})\;\Varid{s}{}\<[E]%
\\
\>[4]{}\mid \Conid{Choice}\;\anonymous \;\Varid{e1}\;\Varid{e2}\Rightarrow {}\<[E]%
\\
\>[4]{}\hsindent{6}{}\<[10]%
\>[10]{}\ensuremath{\mathbf{match}}\;\Varid{parse}\;\Varid{e1}\;\Varid{s}\;\ensuremath{\mathbf{with}}{}\<[E]%
\\
\>[4]{}\hsindent{6}{}\<[10]%
\>[10]{}\mid \Conid{PR\char95 ok}\;\Varid{s'}\;\Varid{v}\Rightarrow \Conid{Ok}\;\Varid{s'}\;\Varid{v}{}\<[E]%
\\
\>[4]{}\hsindent{6}{}\<[10]%
\>[10]{}\mid \Conid{PR\char95 fail}\Rightarrow \Varid{parse}\;\Varid{e2}\;\Varid{s}{}\<[E]%
\\
\>[4]{}\hsindent{6}{}\<[10]%
\>[10]{}\ensuremath{\mathbf{end}}{}\<[E]%
\\
\>[4]{}\mid \Conid{Star}\;\anonymous \;\Varid{e}\Rightarrow {}\<[E]%
\\
\>[4]{}\hsindent{6}{}\<[10]%
\>[10]{}\ensuremath{\mathbf{match}}\;\Varid{parse}\;\Varid{e}\;\Varid{s}\;\ensuremath{\mathbf{with}}{}\<[E]%
\\
\>[4]{}\hsindent{6}{}\<[10]%
\>[10]{}\mid \Conid{PR\char95 fail}\Rightarrow \Conid{Ok}\;\Varid{s}\;[\mskip1.5mu \mskip1.5mu]{}\<[E]%
\\
\>[4]{}\hsindent{6}{}\<[10]%
\>[10]{}\mid \Conid{PR\char95 ok}\;\Varid{s'}\;\Varid{v}\Rightarrow {}\<[E]%
\\
\>[10]{}\hsindent{6}{}\<[16]%
\>[16]{}\ensuremath{\mathbf{match}}\;\Varid{parse}\;(\Varid{e}\;[\mskip1.5mu \mathbin{*}\mskip1.5mu])\;\Varid{s'}\;\ensuremath{\mathbf{with}}{}\<[E]%
\\
\>[10]{}\hsindent{6}{}\<[16]%
\>[16]{}\mid \Conid{PR\char95 fail}\Rightarrow \mathbin{!}{}\<[E]%
\\
\>[10]{}\hsindent{6}{}\<[16]%
\>[16]{}\mid \Conid{PR\char95 ok}\;\Varid{s''}\;\Varid{v'}\Rightarrow \Conid{Ok}\;\Varid{s''}\;(\Varid{v}\mathbin{::}\Varid{v'}){}\<[E]%
\\
\>[10]{}\hsindent{6}{}\<[16]%
\>[16]{}\ensuremath{\mathbf{end}}{}\<[E]%
\\
\>[4]{}\hsindent{6}{}\<[10]%
\>[10]{}\ensuremath{\mathbf{end}}{}\<[E]%
\\
\>[4]{}\mid \Conid{Not}\;\anonymous \;\Varid{e}\Rightarrow {}\<[E]%
\\
\>[4]{}\hsindent{6}{}\<[10]%
\>[10]{}\ensuremath{\mathbf{match}}\;\Varid{parse}\;\Varid{e}\;\Varid{s}\;\ensuremath{\mathbf{with}}{}\<[E]%
\\
\>[4]{}\hsindent{6}{}\<[10]%
\>[10]{}\mid \Conid{PR\char95 ok}\;\anonymous \;\anonymous \Rightarrow \Conid{Fail}{}\<[E]%
\\
\>[4]{}\hsindent{6}{}\<[10]%
\>[10]{}\mid \Conid{PR\char95 fail}\Rightarrow \Conid{Ok}\;\Varid{s}\;\Conid{I}{}\<[E]%
\\
\>[4]{}\hsindent{6}{}\<[10]%
\>[10]{}\ensuremath{\mathbf{end}}{}\<[E]%
\\
\>[4]{}\mid \Conid{Action}\;\anonymous \;\anonymous \;\Varid{e}\;\Varid{f}\Rightarrow {}\<[E]%
\\
\>[4]{}\hsindent{6}{}\<[10]%
\>[10]{}\ensuremath{\mathbf{match}}\;\Varid{parse}\;\Varid{e}\;\Varid{s}\;\ensuremath{\mathbf{with}}{}\<[E]%
\\
\>[4]{}\hsindent{6}{}\<[10]%
\>[10]{}\mid \Conid{PR\char95 ok}\;\Varid{s'}\;\Varid{v}\Rightarrow \Conid{Ok}\;\Varid{s'}\;(\Varid{f}\;\Varid{v}){}\<[E]%
\\
\>[4]{}\hsindent{6}{}\<[10]%
\>[10]{}\mid \Conid{PR\char95 fail}\Rightarrow \Conid{Fail}{}\<[E]%
\\
\>[4]{}\hsindent{6}{}\<[10]%
\>[10]{}\ensuremath{\mathbf{end}}{}\<[E]%
\\
\>[4]{}\mid \mathbin{...}{}\<[E]%
\\
\>[4]{}\ensuremath{\mathbf{end}}{}\<[E]%
\ColumnHook
\end{hscode}\resethooks
The termination argument for this function is based on the decrease 
of the pair of arguments \ensuremath{(\Varid{e},\Varid{s})} in recursive calls with respect to the following
relation $\succ$:
\begin{multline*}
  {(e_1, s_1) \succ (e_2, s_2)} \quad{\iff}\quad
  {
  \exists_{n_1, r_1, n_2, r_2}\ 
    {\pegsemx{e_1}{s_1}{n_1}{r_1}} \land 
    {\pegsemx{e_2}{s_2}{n_2}{r_2}} \land 
    {n_1 > n_2}
  }
\end{multline*}
So $(e_1, s_1)$ is bigger than $(e_2, s_2)$ in the order if its
step-count in the semantics is bigger. The relation $\succ$ is
clearly well-founded, due to the last conjunct with~$>$, the 
well-founded order on $\nat$. Since the semantics
of $\GG$ is complete (due to Theorem~\ref{th:pegs-wf-total} and
the check for well-formedness of $\GG$ as described in 
Section~\ref{sec:trx-wf}) we can prove that all recursive calls
are indeed decreasing with respect to $\succ$.

Clearly this function also generates a number of proof obligations for
expressing correctness of the returned result with respect to
the semantics of PEGs. Dismissing them is actually rather straightforward,
due to the fact that the implementation of the interpreter and the
operation semantics of PEGs are very close to each other. That
means that \emph{by far the majority of our work was in establishing
termination, not correctness}.
\section{Extracting a Parser: Practical Evaluation}\label{sec:pegs-ex}

In the previous section we described a formal development of an XPEG interpreter in the proof assistant \coq. This should allow us for an arbitrary, well-formed XPEG $\GG$, to specify it in \coq\ and, using \coq's extraction capabilities \cite{Let08}, to obtain a certified parser for $\GG$. 
We are interested in code extraction from \coq, to ease practical use of TRX and to improve its performance.
At the moment target languages for extraction from \coq\ are OCaml~\cite{ocaml}, Haskell~\cite{haskell} and 
Scheme~\cite{SusSte98}.
We use the FSets~\cite{FilLet04} library (part of the \coq\ standard library for manipulation of the set data-type) developed using \coq's modules and functors~\cite{Chr03}, which are not yet supported by extraction to Haskell or Scheme.
However, there is an ongoing work on porting FSets to type classes~\cite{SozOur08},
which are supported by extraction.

First, in Section~\ref{sec:eval-improve}, we will sketch the various 
performance-related improvements that we made along our development and 
present case studies on two examples: XML and Java. Then
in Section~\ref{sec:eval-cmp} we will present a benchmark of certified TRX
again a number of other tools on those two examples.


\subsection{Case study of TRX on XML and Java}\label{sec:eval-improve}

A well-known issue with extraction is the performance of obtained programs~\cite{CruLet06,Let08}. Often the root of this problem is the fact that many formalizations are not developed with extraction in mind and trying to extract a computational part of the proof can easily lead to disastrous performance \cite{CruLet06}. On the other hand the CompCert project~\cite{Ler09} is a well-known example of extracting a certified compiler with satisfactory performance from a \coq\ formalization.

As most of TRX's formalization deals with grammar well-formedness, which should be discarded in the extracted code, we aimed at comparable performance for certified TRX and its non-certified counterpart that we prototyped manually. We found however that the first version's performance was unacceptable and required several improvements, which we will discuss in the remainder of this section.

We started with a case study of XML using an XML PEG developed internally at MLstate.
The first extracted version of TRX-cert parsed 32kB of XML in more than one minute.  To our big surprise, performance was somewhere between quadratic and cubic with rather large constants.  To our even bigger surprise, inspection of the code revealed that the \coqid{rev} function from \coq's standard library (from the module \coqid{Coq.Lists.List}) that reverses a list was the source of the problem. The \coqid{rev} function is implemented using \coqid{append} to concatenate lists at every step, hence yielding quadratic time complexity.

We used this function to convert the input from OCaml strings to the extracted type of \coq\ strings. This is another difficulty of working with extracted programs: all the data-types in the extracted program are defined from scratch and combining such programs with un-certified code, even just to add a minimal front-end, as in our case, sometimes requires translating back and forth between OCaml's primitive types and the extracted types of \coq.

Fixing the problem with \coqid{rev} resulted in a linear complexity but the constant was still unsatisfactory. We quickly realized that implementing the range operator by means of repeated choice is suboptimal as a common class of letters $\prange{a}{z}$ would lead to a composition of 26 choices. Hence we extended the semantics of XPEGs with semantics of the range operator and instead of deriving it implemented it ``natively''.

Yet another surprise was in store for us as the performance instead of improving got worse by approximately $30\%$. This time the problem was the fact that in \coq\ there is no predefined polymorphic comparison operator (as in OCaml) so for the range operation we had to implement comparison on characters. We did that by using the predefined function from the standard library converting a character to its ASCII code. And yet again we encountered a problem that the standard library is much better suited for reasoning than computing: this conversion function uses natural numbers in Peano representation. By re-implementing this function using natural numbers in binary notation (available in the standard library) we decreased the running time by a factor of $2$.

Further profiling the OCaml program revealed that it spends $85\%$ of its time performing garbage collection (GC). By tweaking the parameters of OCaml's GC, we obtained an important $3$x gain, leading to TRX-cert's current performance as presented in the following section. We believe a more careful inspection will reveal more potential sources of improvements, as there is still a gap between the performance that we reached now and the one of our prototype written by hand.

We continued with a more realistic case study based on parsing the Java language,
using the PEG for Java developed by Redziejowski~\cite{Red07}. The grammar, consisting 
of 216 rules, was automatically translated to TRX format. We immediately hit performance 
problems as our encoding contains a type enumerating all the rules (\coqid{prod}) and 
proving that equality is decidable on this type, using Coq's \coqid{decide equality} tactic, 
took initially $927$ sec. ($\approx 15$ minutes). We were able to improve it
by writing a tactic dedicated to such simple enumeration types (using Coq's Ltac language)
and decrease this time to $104$ sec. 

We did not meet any more scaling difficulties. Testing XML and Java
grammars for well-formedness, with the extracted Ocaml code, took, respectively,
$0.1$ and $0.7$ sec. (this test needs to be performed only once). We will
discuss the performance of the parsing itself, and compare it with other tools,
in the following section.

\subsection{Performance comparison}\label{sec:eval-cmp}

For our benchmarking experiment, see Figure~\ref{fig:trx-perf} on the following page, we used the following tools:
\begin{enumerate}[\hbox to8 pt{\hfill}]                                        
 \item\noindent{\hskip-12 pt\bf JAXP:}\ a reference implementation for the XML parser, using a DOM parser
   of the ``Java API for XML processing'', JAXP~\cite{JAXP}. 
 \item\noindent{\hskip-12 pt\bf JavaCC:}\ a Java parser~\cite{JavaParser} written in Java using 
   JavaCC~\cite{JavaCC} parser generator.
 \item\noindent{\hskip-12 pt\bf TRX-cert:}\ the certified TRX interpreter, which is the subject of
   this paper and is described in more detail in Section~\ref{sec:pegs-int}.
 \item\noindent{\hskip-12 pt\bf TRX-gen:}\ MLstate's own production-used PEG-based parser generator (for 
   experiments we used its simple version without memoization).
 \item\noindent{\hskip-12 pt\bf TRX-int:}\ a simple prototype with comparable functionality to
   TRX-cert, though developed manually. 
 \item\noindent{\hskip-12 pt\bf Mouse:}\ a PEG-based parser generator, with no memoization, implemented 
   in Java by Redziejowski~\cite{Mouse}.
\end{enumerate}

\tikzstyle barchart=[fill=black!20,draw=black]
\tikzstyle baremph=[fill=black!40,draw=black]
\tikzstyle barscale=[very thin,draw=black!75]

\newcommand{\score}[2]{
  \begin{minipage}[c]{5cm}
  \begin{tikzpicture}
    \draw (0cm,0cm) (5,0.5);
    \draw[#1] (0,0.10) rectangle (#2, 0.40);
  \end{tikzpicture}
  \end{minipage}
}
\begin{figure}[t!]
\begin{center}
\begin{tabular}{| l | r@{.}l@{s.\,} c | r@{.}l@{s.\,} c |}
\toprule
\multicolumn{1}{|c|}{tool} & \multicolumn{3}{c|}{XML parser} & \multicolumn{3}{c|}{Java parser} \\
\midrule
JAXP       &   2&3 & \score{barchart}{0.06} &  \multicolumn{3}{c|}{}         \\
JavaCC     & \multicolumn{3}{c|}{}          &  23&0 & \score{barchart}{0.17} \\
TRX-gen    &   5&1 & \score{barchart}{0.12} &  25&5 & \score{barchart}{0.19} \\
TRX-int    &  40&0 & \score{barchart}{0.97} & 289&3 & \score{barchart}{2.18} \\
TRX-cert   & 128&9 & \score{baremph} {3.12} & 662&4 & \score{baremph} {5.00} \\
Mouse      & 206&4 & \score{barchart}{5.00} & 269&6 & \score{barchart}{2.04} \\
\bottomrule
\end{tabular}
\end{center}
\caption{Performance of certified TRX (TRX-cert) compared to a number of other tools on the examples of parsing 
Java and XML.}
\label{fig:trx-perf}
\end{figure}

Figure~\ref{fig:trx-perf} plots performance of the aforementioned tools on two benchmarks:
\begin{enumerate}[\hbox to8 pt{\hfill}]    
 \item\noindent{\hskip-12 pt\bf XML:}\ 10 XML files with a total size of 40MB generated using 
   the XML benchmarks generator XMark~\cite{XMark}. 
 \item\noindent{\hskip-12 pt\bf Java:}\ a complete source code of the J2SE JDK 5.0 consisting of
   nearly 11.000 files with a total size of 117MB. 
\end{enumerate}\smallskip

\noindent The most interesting comparison is between TRX-cert and TRX-int. The latter was essentially a prototype of the former but developed manually, whereas TRX-cert is extracted from a formal \coq\ development. At the moment the certified version is approximately $2-3$x slower. In principle this difference can be attributed either to the verification overhead (computations that are but should not be performed, as they are part of the logical reasoning to prove correctness and not of the actual algorithm), extraction overhead (sub-optimal code generated by the extraction process) or algorithmic overhead (the algorithm that we coded in \coq\ is sub-optimal in itself).

We believe there is no \emph{verification overhead} in TRX-cert, as all the correctness proofs are discarded by the process of extraction and we never used the proof mode of \coq\ to define objects with computational content (which are extracted).

The \emph{extraction overhead} in our case mainly manifests itself in many dispensable conversions. For instance the second component of the sigma type \ensuremath{\hspace{0.9mm}\{\Varid{x}\mathbin{:}\Conid{T}\mid \Conid{P}\;(\Varid{x})\mskip1.5mu\}} is discarded during the extraction, so such a type is extracted simply as \ensuremath{\Conid{T}} and the first projection function \ensuremath{\Varid{proj1\char95 sig}} as identity. Since sigma types are used extensively in our verification, the extracted code is full of such vacuous conversions. However, our experiments seem to indicate that Ocaml's compiler is capable of optimizing such code, so that this should have no noticeable impact on performance.

Apart from those two types of overheads associated with extraction, often the sub-optimal extracted code can be tracked back to sub-optimal code in the development itself or in \coq\ libraries. We already mentioned few of such problems in Section~\ref{sec:eval-improve}. We believe another one is the model of characters from the standard library of \coq, \ensuremath{\Conid{\Conid{Coq}.\Conid{Strings}.Ascii}}, which we used in this work. The characters are modeled by 8 booleans, \ie, 8 bits of the character:\begin{hscode}\SaveRestoreHook
\column{B}{@{}>{\hspre}l<{\hspost}@{}}%
\column{3}{@{}>{\hspre}l<{\hspost}@{}}%
\column{E}{@{}>{\hspre}l<{\hspost}@{}}%
\>[3]{}\ensuremath{\mathbf{Inductive}}\;\Varid{ascii}\mathbin{:}\Conid{Set}\mathbin{:=}\Conid{Ascii}\;(\anonymous \;\anonymous \;\anonymous \;\anonymous \;\anonymous \;\anonymous \;\anonymous \;\anonymous \mathbin{:}\Varid{bool}).{}\<[E]%
\ColumnHook
\end{hscode}\resethooks
Not surprisingly such characters induce larger memory footprint and also comparison between such structures is much less efficient than between native (1-byte) characters of Ocaml. There is an on-going work on improving interplay between Ocaml's native types and their \coq\ counter-parts, which should hopefully address this problem.
\bigskip

However, the main opportunity for improving performance seems to be in switching
from interpretation to \emph{code generation}. As witnessed by the difference between
TRX-int and TRX-gen this can have a very substantial impact on performance. We will
say some more about that in discussion in Section~\ref{sec:discussion}.


It is worth noting that the performance of TRX-cert is quite competitive when compared
with Java code generated by Mouse.


We would like to conclude this section with the observation that even though 
making such benchmarks is important it is often just one of many factors for
choosing a proper tool for a given task. There are many applications which
will never parse files exceeding $100kB$ and it is often irrelevant whether
that will take $0.1s.$ or $0.01s.$ For some of those applications it may
be much more relevant that the parsing is formally guaranteed to be correct.
\section{Related Work}\label{sec:relwork}

Parsing is a well-studied and well-understood topic and the software
for parsing, parser generators or libraries of parser combinators,
is abundant. And yet there does seem to be hardly any work on
\emph{formally verified} parsing.

Danielsson~\cite{Dan10} develops a library of parser combinators
(see Hutton \cite{Hut92}) with termination guarantees in the dependently
typed functional programming language Agda~\cite{Agda} (see also
joined work with Norell~\cite{DanNor08}). The main difference in comparison with
our work is that Danielsson provides a library of combinators, whereas we
aim at a parser generator for PEG grammars (though at the moment we
only have an interpreter). Perhaps more importantly, the approach of
Danielsson allows many forms of left recursion, which we cannot handle
at present. Another difference is in the way termination is ensured: Danielsson
uses dependent types to extend type of parser combinators with the
information about whether or not they accept the empty string; which
is subsequently used to guarantee termination. In contrast we use
deep embedding of the grammar and a reflective procedure to check whether
a given grammar is terminating. Some consequences of those choices will
be explored in more depth in the following section.

Ideas similar to Danielsson and Norell \cite{DanNor08} were previously put forward, 
though just as a proof of concept, by McBride and McKinna \cite{BriKin02}.

Probably the closest work to ours is that of Barthwal and Norrish~\cite{BarNor09}, where the authors 
developed an SLR parser in HOL. The main differences with our work are:
\begin{iteMize}{$\bullet$}
\item PEGs are more expressive that SLR grammars, which are
  usually not adequate for real-world computer languages,
\item as a consequence of using PEGs we can deal with lexical analysis, while it would have 
  to be formalized and verified in a separate stage for the SLR approach.
\item our parser is proven to be totally correct, \ie, correct with respect to its specification 
  and terminating on all possible inputs (which was actually far more difficult to establish 
  than correctness), while the latter property does not hold for the work of Barthwal and Norrish.
 \item performance comparison with this work is not possible as the paper does not present 
   any case-studies, benchmarks or examples, but the fact that ``the DFA states are computed on the fly'' 
   \cite{BarNor09} suggests that the performance was not the utmost goal of that work.
\end{iteMize}\smallskip

\noindent Finally there is the recent development of a packrat PEG
parser in Coq by Wisnesky et al. \cite{WisEA}, where the given PEG
grammar is compiled into an imperative computation within the Ynot
framework, that when run over an arbitrary imperative character
stream, returns a parsing result conforming with the specification of
PEGs. Termination of such generated parsers is not guaranteed.

\section{Discussion and Future Work}\label{sec:discussion}

One of the main challenges in developing a certified parser is ensuring its termination.
In this paper we presented an extrinsic approach to this problem: we use a deep embedding 
to represent parsing expressions in Coq and then develop a certified algorithm to
verify that a given PEG is well-formed. We then express the parser (interpreter) with
non-structural recursion and the well-formedness of the grammar allows us to justify
that the recursion is well-founded.

There is an alternative, intrinsic approach to the problem of termination, which is, 
for instance, used by Danielsson \cite{DanNor08, Dan10}, as mentioned in the
previous section. They develop a library of parser combinators and use the type
system of the host language -- in this case, Agda -- to restrict the parser combinators
to well-formed ones.

This is a very attractive approach, as by cleverly using the type system of the host
language we obtain certain verified properties for free, hence decreasing the formalization overhead.
However, it has the usual drawback of a shallow embedding approach: it is tied to
the host language, i.e. Danielsson's parsers must unavoidably be written
in Agda. 

At the moment the same is true about our work: to use certified TRX, as presented
in this paper, the grammar must be expressed in Coq. However, this is not a necessity 
with our approach, as we will sketch in a moment. The motivation for avoiding the need 
to use Coq is clear: this could make our certified parser technology usable for people 
outside of the small community of theorem provers (Coq, in particular) experts.

As our work uses deep embedding of parsing expressions, it should be possible to turn
it into a \textit{generic parser generator}.
Doing so could be accomplished by \textit{bootstrapping} TRX: it should be possible to 
write a grammar in it that would synthesize a PEG in Coq (in our format; 
Section~\ref{sec:trx-spec-peg}) from its textual description. After this transformation 
the grammar could be checked for well-formedness (with our generic procedure for
checking well-formedness of PEGs; Section~\ref{sec:trx-wf}) finally allowing parsing
with this grammar (with our interpreter; Section~\ref{sec:trx-int}). This would 
result (via extraction) in a tool that would be capable of parsing grammars expressed 
in a simple textual markup, hence surpassing any need to use/know Coq for the users of 
such a tool.

The main difficulty with obtaining such a tool lies in the bootstrapping process. To do
so we would need a kind of a higher-order grammar: a PEG formally describing its own syntax,
that would take a textual description of a grammar and turn it into a PEG in our format.
Such a grammar would need to have the type \ensuremath{\Conid{PExp}\;(\Conid{PExp}\;(\anonymous ))} and, as already hinted in
Section~\ref{sec:trx-spec-peg}, with our present encoding, that would lead to universe
inconsistency problems. Also, our current use of module system precludes such use-case
as modules are not first-class citizens in \coq\ and one cannot construct higher-order
functors.

But there is a more fundamental problem here: how do we synthesize semantic actions
from their textual description? If the semantics actions were to be expressed in
the calculus of constructions of \coq, the way they are now, this seems to be futile.

Let us step back a bit for a moment and consider a simpler problem: what if we only
wanted a \emph{recognizer}, \ie, a parser that does not return any result, but only
indicates whether a given string is in the language described by the grammar or not.
To address the aforementioned problem with modules (\cite{Chr03}) we could switch to type
classes (\cite{SozOur08}) instead. Then we could build a generic recognizer
as follows (pseudo-code):
\begin{hscode}\SaveRestoreHook
\column{B}{@{}>{\hspre}l<{\hspost}@{}}%
\column{3}{@{}>{\hspre}l<{\hspost}@{}}%
\column{5}{@{}>{\hspre}l<{\hspost}@{}}%
\column{E}{@{}>{\hspre}l<{\hspost}@{}}%
\>[3]{}\ensuremath{\mathbf{Definition}}\;\Conid{PEG\char95 grammar}\mathbin{:}\Conid{PExp}\;\Varid{pexp}\mathbin{:=}\mathbin{...}{}\<[E]%
\\[\blanklineskip]%
\>[3]{}\ensuremath{\mathbf{Program}}\;\ensuremath{\mathbf{Definition}}\;\Varid{do\char95 parse}\;(\Varid{grammar}\mathbin{:}\Varid{string})\;(\Varid{input}\mathbin{:}\Varid{string})\mathbin{:=}{}\<[E]%
\\
\>[3]{}\hsindent{2}{}\<[5]%
\>[5]{}\ensuremath{\mathbf{match}}\;\Varid{parse}\;\Conid{PEG\char95 grammar}\;\Varid{grammar}\;\ensuremath{\mathbf{with}}{}\<[E]%
\\
\>[3]{}\hsindent{2}{}\<[5]%
\>[5]{}\mid \Conid{PR\char95 ok}\;\anonymous \;\Varid{peg}\Rightarrow \Varid{parse}\;(\Varid{promote}\;\Varid{peg})\;\Varid{input}{}\<[E]%
\\
\>[3]{}\hsindent{2}{}\<[5]%
\>[5]{}\mid \Conid{PR\char95 fail}\Rightarrow \Conid{PR\char95 fail}{}\<[E]%
\\
\>[3]{}\hsindent{2}{}\<[5]%
\>[5]{}\ensuremath{\mathbf{end}}.{}\<[E]%
\ColumnHook
\end{hscode}\resethooks
Here \ensuremath{\Conid{PEG\char95 grammar}} is the grammar for PEGs. The main \ensuremath{\Varid{do\char95 parse}} function takes
two arguments: \ensuremath{\Varid{grammar}} with the textual description of the grammar to use and
\ensuremath{\Varid{input}} being the \ensuremath{\Varid{input}} which we want to parse using the given \ensuremath{\Varid{grammar}}. We use
\ensuremath{\Conid{PEG\char95 grammar}} to parse \ensuremath{\Varid{grammar}} and, hopefully, obtain its internal representation
\ensuremath{\Varid{peg}\mathbin{:}\Varid{pexp}}, in which case we again invoke \ensuremath{\Varid{parse}} with \ensuremath{\Varid{promote}\;\Varid{peg}} grammar and
\ensuremath{\Varid{input}} as the input string. Extracting \ensuremath{\Varid{do\char95 parser}} would give us a generic
recognizer, that could be used without Coq (or any knowledge thereof).

Admittedly, in practice we are rarely interested in merely validating the input;
usually we really want to \textit{parse} it, obtaining its structural representation.
How can the above approach be extended to accommodate that and still result in a
stand-alone tool, not requiring interaction with Coq?

One option would be to move from interpretation to code generation and then using 
the target language to express semantic actions. An additional advantage is that this
should result in a big performance gain (compare the performance of TRX and TRX-int in 
Figure~\ref{fig:trx-perf}). But that would be a major undertaking requiring 
reasoning with respect to the target language's semantics for the correctness proofs
and some sort of (formally verified) termination analysis for that language, to 
ensure termination of the code of semantic actions (and hence the generated parser). 

The aforementioned termination problem for a parser generator could be simplified 
by restricting the code allowed in semantic actions to some subset of the target 
language, which is still expressive enough for this purpose but for which the termination 
analysis is simpler. For instance for a purely functional target language one could
disallow recursion altogether in productions (making termination evident), only allowing 
use of some predefined set of combinators (to improve expressivity of semantic actions), 
which could be proven terminating manually.

Another solution would be not to use semantic actions altogether, but construct a
parse tree, the shape of which could be influenced by annotations in the grammar.
This is the approach used, for instance, in the Ocaml PEG-based parser generator
Aurochs~\cite{aurochs}. We believe this is a promising approach that we hope
to explore in the future work.

\bigskip
A complete different approach to developing a practical, certified parser generator
would be the standard technique of verification \emph{a posteriori}: use an untrusted
parser that, apart from its result, generates some sort of a certificate (parse tree)
and develop a (formally correct) tool to verify, using the certificate, that the output
of the tool (for a given input and given grammar) is correct. The attractiveness of this
approach lies in the fact that such a verifier would typically be much
simpler than the parser itself. There are two problems with this approach though:
\begin{iteMize}{$\bullet$}
 \item this approach could at best give us partial correctness guarantees, as we
   would not be able to ensure termination of the un-trusted parser (unless we
   also prove it in some way);
 \item if the parsing is successul it is relatively clear what a certificate should be
   (parse tree), but what if it is not? How can we certify incorrectness of input
   with respect to the grammar?
\end{iteMize}

\bigskip
Apart from making the certified TRX a Coq independent, standalone tool and moving
from interpretation to code generation we also identify a number of other possible 
improvements to TRX as future work:

\begin{enumerate}[(1)]
 \item\label{future:performance} 
   Linear parsing time with PEGs can be ensured by using packrat parsing~\cite{For02},
   \ie, enhancing the parser with memoization. This should be relatively easy to 
   implement (it has, respectively, no and little impact on the termination and 
   correctness arguments for certified TRX), but induces high memory costs
   (and some performance overhead), so it is not clear whether this would be 
   beneficial. An alternative would be to develop (formally verified?) tools
   to perform grammar analysis and warn the user in case the grammar can lead to
   exponential parsing times.

 \item Another important aspect is that of left-recursive grammars, which
   occur naturally in practice. At the moment it is the responsibility of
   the user to eliminate left-recursion from a grammar. In the future, we
   plan to address this problem either by means of left-recursion elimination
   \cite{For02mth}, \ie, transforming a left-recursive grammar to an equivalent 
   one where left-recursion does not occur (this is not an easy problem in 
   presence of semantic actions, especially if one also wants to allow mutually 
   left-recursive rules). Another possible approach is an extension to the 
   memoization technique that allows dealing with left-recursive rules~\cite{WarEA08}.

 \item\label{future:errmsg}
   Finally support for \emph{error messages}, for instance following that of
   the PEG-based parser generator Puppy~\cite{For02mth}, would greatly
   improve usability of TRX.
\end{enumerate}
\section{Conclusions}\label{sec:concl}

In this paper we described a \coq\ formalization of the theory of PEGs and,
based on it, a formal development of \emph{\TRX: a formally verified parser 
interpreter for PEGs}. This allows us to write a PEG, together with its 
semantic actions, in \coq\ and then to extract from it a \emph{parser with 
total correctness guarantees}. That means that the parser will terminate 
on all inputs and produce parsing results correct with respect to the 
semantics of PEGs. 
Considering the importance of parsing, this result appears as a first step 
towards a general way to bring added quality and security to all kinds of software .

The emphasis of our work was on \emph{practicality}, so apart from treating
this as an interesting academic exercise, we were aiming at obtaining a 
tool that scales and can be applied to real-life problems. We performed
a case study with a (complete) Java grammar and demonstrated that the resulting 
parser exhibits a reasonable performance. We also stressed the importance
of making those results available to people outside of the small circle 
of theorem-proving experts and presented a plan of doing so as future
work.


\subsection*{Acknowledgments}
We would like to thank Matthieu Sozeau for his invaluable help with the
Program feature~\cite{Soz07} of \coq\ and the anonymous referees for their
helpful comments, which greatly improved presentation of this paper.
Also the very pragmatic (and immensely helpful) book of Chlipala~\cite{Chl09},
as well as friendly advice from people on Coq's mailing list turned out to
be invaluable in the course of this work.


\bibliographystyle{alpha}
\bibliography{paper}

\end{document}